\documentclass[11pt,peerreview,draftcls,onecolumn]{IEEEtran}

\newif\ifsubmission
\submissionfalse 

\usepackage[latin1]{inputenc}
\usepackage{amsmath}
\usepackage{amsfonts}
\usepackage{amssymb}
\usepackage{amstext}
\usepackage{amsthm}
\usepackage{graphicx}
\usepackage{url}
\usepackage{color}
\usepackage{xcolor}
\usepackage{bm}
\usepackage{dsfont}

\newtheorem{theorem}{Theorem}
\newtheorem{lemma}{Lemma}

\newtheorem{definition}{Definition}
\newtheorem{property}{Property}

\newcommand{\A}{\mathcal A}

\newcommand{\DD}{\mathcal D}
\newcommand{\XX}{\mathcal X}

\renewcommand{\SS}{\mathcal S}
\newcommand{\MM}{\mathcal M}

\newcommand{\PP}{\mathcal P}

\newcommand{\YY}{\mathcal Y}

\newcommand{\rarrow}{\tiny \to}

\colorlet{mygreen}{green!60!gray}

\graphicspath{{../images/}{images/}}
\ifsubmission
    \newcommand{\TODO}[1]{}
    
\else
    \newcommand{\TODO}[1]{\textcolor{red}{{#1}}}

\fi

\begin{document}

\title{Source Distinguishability under Distortion-Limited Attack: an Optimal Transport Perspective}


\author{Mauro Barni$^*$, \IEEEmembership{Fellow, IEEE}, Benedetta Tondi, \IEEEmembership{Student Member, IEEE}
\thanks{M. Barni and B. Tondi are with the Department of Information Engineering and Mathematical Sciences, University of Siena, Via Roma 56, 53100 - Siena, ITALY, phone: +39 0577 234850 (int. 1005), e-mail: \{barni@dii.unisi.it, benedettatondi@gmail.com\}.}
}

\markboth{IEEE TRANSACTIONS ON INFORMATION THEORY, ~Vol.~X, No.~X, XXXXXXX~XXXX}%
{M. Barni, B. Tondi: Source Distinguishability under Distortion-Limited Attack: an Optimal Transport Perspective}

\maketitle

\begin{abstract}
We analyze the distinguishability of two sources in a Neyman-Pearson set-up when an attacker is allowed to modify the output of one of the two sources subject to a distortion constraint. By casting the problem in a game-theoretic framework and by exploiting the parallelism between the attacker's goal and Optimal Transport Theory, we introduce the concept of Security Margin defined as the maximum average per-sample distortion introduced by the attacker for which the two sources can be distinguished ensuring arbitrarily small, yet positive, error exponents for type I and type II error probabilities. Several versions of the problem are considered according to the available knowledge about the sources and the type of distance used to define the distortion constraint. We compute the security margin for some classes of sources and derive a general upper bound assuming that the distortion is measured in terms of the mean square error between the original and the attacked sequence.
\end{abstract}



\IEEEpeerreviewmaketitle

\section{Introduction}

\IEEEPARstart{A}{dversarial} Signal Processing (Adv-SP), sometimes referred to as adversary-aware signal processing, is an emerging research field targeting the study of signal processing techniques explicitly thought to withstand the attacks of one or more adversaries aiming at system failure. Adv-SP methods can be applied to a wide variety of security-oriented applications including 	multimedia forensics, biometrics, digital watermarking, steganography and steganalysis, network intrusion detection, traffic monitoring, video-surveillance, just to mention a few \cite{BarGon13}. Source identification is a common problem in Adv-SP, due to its importance in several applications. In multimedia forensics, for instance, the analyst may want to distinguish which between two sources (e.g. a photo camera and a scanner) generated a given document, or whether a document has undergone a given processing or not. In spam filtering, e-mail messages have to be classified either as spam or authentic messages. In 1-bit watermarking, the detector has to decide whether a document is watermarked or not, while it is the goal of steganalysis to distinguish between cover and stego-images. In yet other situations, the security of a system relies on the capability of distinguishing the profile of malevolent and fair users.

In \cite{BT13}, a game-theoretic framework is proposed to analyze the source identification problem under adversarial conditions. To be specific, \cite{BT13} introduces the so called source identification game. The game is played by a Defender (D) and an Attacker (A) and is defined as follows: given two discrete memoryless sources $X$ and $Y$ with alphabet $\XX$ and probability mass functions (pmf) $P_X$ and $P_Y$, and a test sequence $x^n = (x_1, x_2 \dots x_n)$, the goal of D is to decide between hypothesis $H_0$ that $x^n$ has been drawn from $X$ and hypothesis $H_1$ that $x^n$ has been generated by $Y$. The goal of A is to take a sequence $y^n$ generated by $Y$ and modify it in such a way that D classifies it as being generated by $X$. In doing so, D must ensure that the type I error probability (usually referred to as false positive error probability $P_{fp}$) of deciding for $H_1$ when $H_0$ holds stays below a given threshold, whereas $A$ has to respect a distortion constraint, limiting the amount of modifications he can introduce into $y^n$. The payoff of the game is the type II error probability, or false negative error probability $P_{fn}$, i.e., the probability of deciding for $H_0$ when $H_1$ holds. Of course, D aims at minimizing $P_{fn}$, while A wishes to maximize it. The above scenario accounts for a situation in which $P_X$ corresponds to so-to-say normal conditions and $P_Y$ refers to an anomalous situation. It is the goal of the attacker to modify a sequence produced under anomalous conditions in such a way that the defender does not recognize that the observed system exited the normal state.

The analysis provided in \cite{BT13} assumes that the defender is confined to base its analysis only on first order statistics of $x^n$. Under this assumption, \cite{BT13} derives the asymptotic equilibrium point of the game when the length of the test sequence tends to infinity and the false positive error probability is required to tend to zero exponentially fast with decay rate at least equal to $\lambda$ ($\lambda$ is nothing but the error exponent of the false positive error probability). Given two pmf's $P_X$ and $P_Y$, a false positive error exponent $\lambda$, and the maximum allowed distortion $L_{max}$, the analysis in \cite{BT13} permits to determine whether, at the equilibrium, the false negative error probability $P_{fn}$ tends to 0 or to 1 when $n \rarrow \infty$. This, in turn, permits to define the so-called indistinguishability region $\Gamma(P_X,\lambda,L_{max})$ as the set of pmf's that can not be distinguished reliably from $P_X$ when $n \rarrow \infty$ due to the presence of the attacker. If $P_Y \in \Gamma(P_X,\lambda,L_{max})$, in fact, a strictly positive false negative error exponent can not be achieved and the attacker is going to {\em win} the game. A similar analysis is carried out in \cite{BT12,BTtit} for a scenario in which $P_X$ and $P_Y$ are not known, and the statistics of the two sources are obtained through the observation of training sequences.

\subsection{Contribution}

A drawback with the analysis carried out in \cite{BT13, BT12, BTtit} is the asymmetric role of the false positive and false negative error exponents, namely $\lambda$ and $\varepsilon$ ($\varepsilon = \lim_{n \rarrow \infty} - \frac{1}{n} \log P_{fn}$). In such works, in fact, the defender aims at ensuring a given $\lambda$, but is satisfied with any strictly positive $\varepsilon$. In this paper, we make a more reasonable assumption and say that the defender wins the game, i.e. he is able to distinguish between $X$ and $Y$ despite the presence of the adversary, if - at the equilibrium - both error probabilities tend to zero exponentially fast, regardless of the particular values assumed by the error exponents. More precisely, by mimicking Stein's lemma \cite{CandT}, we analyze the behavior of $\Gamma(P_X,\lambda,L_{max})$ when $\lambda \rarrow 0$ to see whether, given a maximum allowable distortion $L_{max}$, it is possible for D to simultaneously attain strictly positive error exponents for the two kinds of error, hence permitting to reliably distinguish between $P_X$ and $P_Y$. Having done so, we will adopt a different perspective and introduce a new distinguishability measure, called Security Margin ($\SS\MM$), defined as the {\em maximum distortion allowed to the attacker, for which two sources can be reliably distinguished}. As we will see, this is a powerful concept that permits to summarize in a single quantity the distinguishability of two sources $X$ and $Y$ under adversarial conditions. In order to derive our main results, we look at the optimum attacker's strategy already derived in \cite{BT13} and \cite{BTtit} from a new perspective, i.e. by paralleling it to {\em optimal transport theory} \cite{Vill09}. Doing so, in fact, allows to derive a very intuitive and insightful interpretation of the optimum attacker's strategy, and will permit us to derive the $\SS\MM$ for a wide class of pmf's in both the discrete and the continuous case. A fast numerical algorithm for the computation of the security margin between any two discrete pmf's will also be presented.

In the framework depicted above, the main results proven in this paper can be summarized as follows

\begin{enumerate}
\item{We compute the best achievable false negative error exponent for a given distortion $L_{max}$ and for a strictly positive, yet arbitrarily small, value of the false positive error exponent $\lambda$ (Theorem \ref{theorem_EMD}, Section \ref{subsec.SMdef});}
\item{We introduce the security margin ($\SS\MM$) concept as the maximum allowed distortion for which two sources $X$ and $Y$ can be distinguished (in the adversarial setup defined in the paper) by ensuring strictly positive error exponents of the two kinds and show that $\SS\MM$ corresponds to the Earth Mover Distance ({\em EMD}) between $P_X$ and $P_Y$ (Definition \ref{def.SMdef}, Section \ref{subsec.SMdef});}
\item{We extend  the analysis to a version of the source identification game in which $P_X$ and $P_Y$ are known only through training sequences (Source Identification game with training data, $SI_{tr}$) and show that the security margin does not change despite the fact that the $SI_{tr}$ is in general more favorable to the attacker than the $SI_{ks}$ game where the exact statistics of the sources are perfectly known to D and A (Theorem \ref{theorem_EMD_tr}, Section \ref{subsec.SMTR});}
\item{By relying on some results in the field of optimal transport theory, we present a number of ways whereby the $\SS\MM$ can be computed efficiently for both discrete and continuous sources (Section \ref{sec.SMcomputation});}
\item{We introduce a new version of the game in which the distortion constraint is expressed in terms of maximum absolute distance between the sequence $y^n$ and the attacked sequence $z^n$ (Theorem \ref{theorem_EMD_Linfty}, Section \ref{sec.maxdiff}). We then extend our analysis to the new version of the game. This is a very interesting, yet not trivial, scenario, since in many practical applications the quality of the attacked sequence is judged in terms of maximum distance ($L_{\infty}$ norm), rather than in terms of average distance.}
\end{enumerate}

It is worth stressing that point 4) complements and generalizes some recent studies in the field of Multimedia security, namely \cite{FernPedro2013,Balado} regarding image counterforensics, and \cite{balado2013permutation} related to perfect steganography. As a matter of fact, all the solutions proposed in those papers can be seen as particular instances of the general optimal transport problem addressed (and solved) in Section \ref{sec.SMcomputation}. Finally, point 5) relies on a generalization of the results proven in \cite{BT13, BT12,BTtit}, where the analysis was restricted to the case of additive distortion measures. As we will show, such an analysis can be extended to the case of $L_{\infty}$ distortion, opening the way to the application of our methodology to all the scenarios in which the distortion constraint is applied uniformly to the elements of $y^n$.

Some of the results presented in this paper have already been stated in \cite{BTGsip}. With respect to \cite{BTGsip}, however, the current paper contains a complete proof of all the main theorems, the extension to the case of source identification with training data, the derivation of a fast numerical methodology to compute the security margin between any two discrete sources, and the extension of the analysis to the case of $L_{\infty}$ distortion.

The rest of this paper is organized as follows. In Section \ref{sec.not}, we introduce the notation used throughout the paper, give some definitions and review some basic concepts in game theory. In Section \ref{sec.SI_ks}, we give a rigorous definition of the addressed problem and summarize the main results proven in \cite{BT13}. Section \ref{sec.SM} is the core of the paper: we use optimal transport to shed new light on the addressed problem and introduce the security margin concept. In Section \ref{sec.SI_tr}, we extend the analysis to cover the case of source identification with training data. In Section \ref{sec.SMcomputation}, we derive the security margin for several classes of sources, and provide an efficient algorithm to compute it when a close form solution can not be found. Section \ref{sec.maxdiff} extends the analysis to a situation in which the allowed distortion is defined in terms of $L_{\infty}$ distance. The paper ends in Section \ref{sec.conc}, with some conclusions and highlights for future research. The most technical proofs are given in the appendices to avoid interrupting the flow of ideas in the main body of the paper.

\section{Notations and definitions}
\label{sec.not}

In this section we introduce the notation and definitions used throughout the paper. We will use capital letters to indicate discrete memoryless sources (e.g. $X$). Sequences of length $n$ drawn from a source will be indicated with the corresponding lowercase letters (e.g. $x^n$); accordingly, $x_i$ will denote the $i-$th element of a sequence $x^n$. The alphabet of an information source will be indicated by the corresponding calligraphic capital letter (e.g. $\XX$). The probability mass function (pmf) of a discrete memoryless source $X$ will be denoted by $P_X$, while the cumulative mass function will be indicated with $C_X$. For the sake of simplicity, the same notation will be adopted to denote the probability density function (pdf) of a continuous random variable $X$. The calligraphic letter $\PP$ will be used to indicate the class of all the probability density functions. In addition, the notation $P_X$ will be also used to indicate the probability measure ruling the emission of sequences from a source $X$, so we will use the expressions $P_X(a)$ and $P_X(x^n)$ to indicate, respectively, the probability of symbol $a \in \XX$ and the probability that the source $X$ emits the sequence $x^n$, the exact meaning of $P_X$ being always clearly recoverable from the context wherein it is used. Finally, we will use the notation $P_X(A)$ to indicate the probability of the event $A$ (be it a subset of $\XX$ or $\XX^n$) under the probability measure $P_X$.

Our analysis relies extensively on the concepts of type and type class defined as follows (see \cite{CandT} and \cite{CandK} for more details). Let $x^n$ be a sequence with elements belonging to a finite alphabet $\XX$. The type $P_{x^n}$ of $x^n$ is the empirical pmf induced by the sequence $x^n$, i.e. $\forall a \in \XX, P_{x^n} (a) = \frac{1}{n} \sum_{i=1}^n \delta(x_i, a)$, where $\delta(x_i,a) = 1$ if $x_i =a$ and zero otherwise. In the following we indicate with $\PP_n$ the set of types with denominator $n$, i.e. the set of types induced by sequences of length $n$. Given $P \in \PP_n$, we indicate with $T(P)$ the type class of $P$, i.e. the set of all the sequences in $\XX^n$ having type $P$.

The Kullback-Leibler (KL) divergence between two distributions $P$ and $Q$ on the same finite alphabet $\XX$ is defined as:
\begin{equation}
    \DD(P||Q) = \sum_{a \in \XX} P(a) \log \frac{P(a)}{Q(a)},
\label{eq.KLdiv}
\end{equation}
where, according to usual conventions, $0 \log 0 = 0$ and $p \log p/0 = \infty$ if $p > 0$.

\subsection{Game theory in a nutshell}

A 2-player game is defined as a 4-uple $G(\SS_1,\SS_2,u_1, u_2)$, where $\SS_1 = \{s_{1,1} \dots s_{1,n_1}\}$ and $\SS_2 = \{s_{2,1} \dots s_{2,n_2}\}$ are the set of actions (usually called strategies) the first and the second player can choose from, and $u_l(s_{1,i}, s_{2,j}), l= 1,2$, is the payoff of the game for player $l$, when the first player chooses the strategy $s_{1,i}$ and the second chooses $s_{2,j}$. A pair of strategies $(s_{1,i}, s_{2,j})$ is called a profile. When $u_1(s_{s1,i}, s_{2,j}) + u_2(s_{1,i}, s_{2,j}) = 0$, the game is said to be a competitive (or zero-sum) game. In the set-up adopted in this paper, $\SS_1$, $\SS_2$ and the payoff functions are assumed to be known to the two players. In addition, we assume that the players choose their strategies before starting the game without knowing the  strategy chosen by the other player (strategic game).

A common goal in game theory is to determine the existence of equilibrium points, i.e. profiles that in {\em some way} represent a {\em satisfactory} choice for both players \cite{Osb94}. The most famous equilibrium notion is due to Nash. Intuitively, a profile is a Nash equilibrium if each player does not have any interest in changing his choice assuming the other does not change his strategy. Despite its popularity, the practical meaning of Nash equilibrium is doubtful, since there is no guarantee that the players will end up playing at the equilibrium. A notion with a more practical meaning is that of {\em dominant equilibrium}. A strategy is said to be strictly dominant for one player if it is the best strategy for the player, regardless of the strategy chosen by the other player. In many cases dominant strategies do not exist, however when one such strategy exists for one of the players, he will surely adopt it (at least under the assumption of  rational behavior). The other players, in turn, will choose their strategies anticipating that the first player will play the dominant strategy. As a consequence, in a two-player game, if a dominant strategy exists the players have only one rational choice called the only rationalizable equilibrium of the game \cite{ChenGames}. Games with the above property are called {\em dominance solvable} games.

\section{The source identification game with known sources}
\label{sec.SI_ks}


In this section, we give a rigorous definition of the problem considered in the paper. In order to make our treatment self-contained and ease the understanding of subsequent derivations, we also summarize the main results proven in \cite{BT13}. With respect to \cite{BT13}, however, we adopt a different perspective that facilitates the interpretation of the attacker's optimal strategy as the solution of an optimal transport problem. As a matter of fact, this can be considered as an important contribution of this paper, since the new perspective opens the way to the adoption of a new, more insightful, methodology to analyze the structure of the game and  the achievable performance.

\subsection{Definition of the $SI_{ks}$ game and equilibrium point}

We start with the definition of the source identification game with known sources ($SI_{ks}$). Given a test sequence $x^n$, we indicate with $H_0$ the hypothesis that $x^n$ has been generated by $P_X$ and with $H_1$ the alternative hypothesis that $x^n$ has been generated by $P_Y$. In order to define the $SI_{ks}$ game, we need to define the set of strategies of D and A and the payoff function.

{\em Defender's strategies.} The set of strategies of the Defender ($\SS_D$) consists of all possible acceptance regions for $H_0$. More precisely, by following \cite{BT13}, we require that D bases its analysis only on the first order statistics of $x^n$. This is equivalent to ask that the acceptance region for hypothesis $H_0$, hereafter referred to as $\Lambda^n$, is a union of type classes\footnote{We use the superscript $n$ to indicate explicitly that $\Lambda^{n}$ refers to $n$-long sequences.}. Since a type class is univocally defined by the empirical pmf of the sequences it contains, $\Lambda^n$ can be seen as a union of types $P \in \PP_n$. We consider an asymptotic version of the game and require that the false positive error probability $P_{fp}$ decreases exponentially with decay rate at least equal to $\lambda$. Under  the above assumptions, the space of strategies of D is given by:
\begin{equation}
    \SS_{D} = \{ \Lambda^n \in 2^{\PP_n}: P_{fp} \le 2^{-\lambda n}\},
\label{eq.SD_KS_as}
\end{equation}
where $2^{\PP_n}$ indicates the power set of $\PP_n$.

{\em Attacker's strategies.} Given a sequence $y^n$ drawn from $Y$, the goal of A is to transform it into a sequence $z^n$ belonging to the acceptance region chosen by D. Let us indicate by $n(i,j)$ the number of times that the $i$-th symbol of the alphabet is transformed into the $j$-th one as a consequence of the attack. Similarly, we indicate by $S^n_{YZ}(i,j) = n(i,j)/n$ the relative frequency with which the $i$-th symbol of the alphabet is transformed into the $j$-th one. In the following, we refer to $S^n_{YZ}$ as {\em transportation map}. Once again, we explicitly indicate that $S^n_{YZ}$ refers to $n$-long sequences by adding the superscript $n$. For any additive distortion measure, the overall distortion introduced by the attack can be expressed in terms of $n(i,j)$; in fact we have:
\begin{equation}
    d(y^n, z^n) = \sum_{i,j} n(i,j) d(i,j),
\label{eq.overall_dist}
\end{equation}
where $d(i,j)$ is the distortion introduced when the symbol $i$ is transformed into the symbol $j$. Similarly, the average per-sample distortion depends only on $S^n_{YZ}$:
\begin{equation}
    \frac{d(y^n, z^n)}{n} = \sum_{i,j} S^n_{YZ}(i,j) d(i,j).
\label{eq.averagel_dist}
\end{equation}
$S^n_{YZ}$ determines also the empirical pmf (i.e. the type) of the attacked sequence. In fact, by indicating with $P_{z^n}(j)$ the relative frequency of symbol $j$ into $z^n$, we have:
\begin{equation}
\label{eq.out_type}
    P_{z^n}(j) = \sum_i S^n_{YZ}(i,j) \triangleq S^n_Z(j).
\end{equation}
Finally, we observe that the attacker can not change more symbols than there are in the sequence $y^n$; as a consequence a map $S^n_{YZ}$ can be applied to a sequence $y^n$ only if:
\begin{equation}
\label{eq.in_type}
    S^n_{Y}(i) \triangleq \sum_{j} S^n_{YZ}(i,j) = P_{y^n}(i).
\end{equation}
Equations (\ref{eq.out_type}) and (\ref{eq.in_type}) suggest an interesting interpretation of $S^n_{YZ}$, which can be seen as the joint empirical pmf between the sequences $y^n$ and $z^n$. In the same way, $S^n_{Y}$ and $S^n_{Z}$ correspond, respectively, to the empirical pmf of $y^n$ and  $z^n$.

By remembering that $\Lambda^n$ depends only on the empirical pmf of the test sequence (i.e., on its type), and given that the empirical pmf of the attacked sequence depends  on $S^n_Z$ only through $S^n_{YZ}$, we can define the action of the attacker as the choice of a transportation map among all {\em admissible} maps, a map being admissible if:
\begin{align}
\label{eq.admissiblemap1}
    & S^n_{Y} = P_{y^n} \\ \nonumber
    & \sum_{i,j} S^n_{YZ}(i,j) d(i,j) \le L_{max},
\end{align}
where the second condition expresses the per-letter distortion constraint the attacker is subject to, and $L_{max}$ is the maximum allowable (average) per-letter distortion. In the following, we will refer to the set of admissible maps as $\A^n(L_{max}, P_{y^n})$. With the above definitions, the space of strategies of the attacker is the set of all the possible ways of associating an admissible transformation map to the to-be-attacked sequence.  In the following, we will refer to the result of such an association as $S^n_{YZ}(y^n)$, or $S^n_{YZ}(i,j;y^n)$, when we need to refer explicitly to the relative frequency with which the symbol $i$ is transformed into the symbol $j$. In the same way, $S^n_Z(j;y^n)$ indicates the output marginal of $S^n_{YZ}(i,j;y^n)$\footnote{With regard to the input marginal, of course, we always have $S^n_Y(i;y^n) = P_{y^n}(i) ~ \forall i$.}.
By adopting the above symbolism, the space of strategies for the attacker can be defined as:
\begin{equation}
    \SS_{A} = \{ S^n_{YZ}(i,j;y^n) : S^n_{YZ}(i,j)  \in \A^n(L_{max}, P_{y^n}) \}.
\label{eq.SA_TR_as}
\end{equation}
{\em The payoff.} Having fixed the maximum false positive error probability, we adopt a typical Neyman-Pearson approach and let the payoff correspond  to the false negative error probability, that is:
\begin{equation}
\label{equation}
    u_D = -u_A = - \sum_{y^n: S^n_{Z}(j;y^n) \in \Lambda^n} P_Y(y^n),
\end{equation}
where $P_Y(y^n)$ is the probability that the source $Y$ outputs the sequence $y^n$.

{\em Equilibrium point.} Given the above formulation of the $SI_{ks}$ game, the main result of \cite{BT13} is summarized by the following theorem.\footnote{In this paper we use a different formulation of the theorem with respect to \cite{BT13} so to adapt it to the new formalism based on the concept of transportation map adopted here.}
\begin{theorem}
Let
\begin{equation}
    \Lambda^{n,*} = \left\{P \in \PP_n: \DD(P || P_X) < \lambda - |\XX| \frac{\log(n+1)}{n} \right\},
\label{eq.opt_lambda_KS_as}
\end{equation}
and
\begin{equation}
    S_{YZ}^{n,*}(i,j;y^n) = \arg\min_{S^n_{YZ} \in \A^n(L_{max},P_{y^n})} \DD(S^n_{Z} || P_X).
\label{eq.opt_AD_KS_as}
\end{equation}
Then $\Lambda^{n,*}$ is a dominant equilibrium for D and the profile $(\Lambda^{n,*},  S_{YZ}^{n,*}(i,j;y^n))$ is the only rationalizable equilibrium of the $SI_{ks}$ game, which, then, is a dominance solvable game.
\label{theo.NASH_KS_as}
\end{theorem}

\subsection{Payoff of the $SI_{ks}$ game at the equilibrium}

Given the optimal acceptance region $\Lambda^{n,*}$ and the optimum attacking strategy $S_{YZ}^{n,*}(y^n)$, we can introduce the indistinguishability region $\Gamma^n(P_X, \lambda, L_{max})$ as follows:
\begin{align}
\label{eq.indistinguib}
    \Gamma^n(P_X, \lambda, L_{max}) &= \\ \nonumber
    \{P \in \PP_n & : \exists ~ S^n_{YZ} \in \A^n(L_{max}, P) \text{ s.t. } S^n_{Z} \in \Lambda^{n,*}\}.
\end{align}
The indistinguishability region defines all the type classes (with denominator $n$) whose sequences can be moved within  $\Lambda^{n,*}$ by the attacker.
The problem with the above analysis is that it applies only to types with denominator $n$ and hence can not be used to decide whether the sequences generated by two generic sources (not necessarily belonging to $\PP_n$) can be distinguished. In order to answer this question, we can rely on the density of rational numbers in $\mathds{R}$, and let $n$ tend to infinity. In this way we can define the asymptotic counterpart of $\Gamma^n$, specifying whether two sources can eventually be distinguished for increasing values of $n$ \cite{BT13}:
\begin{align}
\label{eq.indistinguib_limit}
    \Gamma(P_X, \lambda, L_{max}) &= \\ \nonumber
    \{P \in \PP : & ~ \exists ~ S_{YZ} \in \A(L_{max}, P) \text{ s.t. } S_{Z} \in \Lambda^*(P_X, \lambda)\},
\end{align}
where
\begin{equation}
\label{eq.lambda_limit}
    \Lambda^*(P_X,\lambda) = \{P \in \PP : \DD(P|| P_X) \le \lambda\},
\end{equation}
and where the definitions of $S_{YZ}(i,j)$, $S_Z(j)$ and $\A(L_{max}, P)$ are obtained immediately from those of $S^n_{YZ}(i,j)$, $S^n_Z(j)$ and  $\A^n(L_{max}, P)$, by relaxing the requirement that $S_{YZ}(i,j)$, $S_Z(j)$ and $P(i)$ are rational numbers with denominator $n$. More precisely,  we can state the following theorem:
\begin{theorem}
\label{theorem2}
For the $SI_{ks}$ game, the error exponent of the false negative error probability at the equilibrium is given by\footnote{Here and in the rest of the paper, the use of the minimum instead of the infimum is justified by the compactness of $\Gamma(P_X, \lambda, L_{max})$ and other similar sets defined in the following.}:
\begin{equation}
    \varepsilon = \min_{P \in \Gamma(P_X, \lambda, L_{max})} \DD(P||P_Y),
\label{eq.asymptotic_theorem}
\end{equation}
leading to the following cases:
\begin{enumerate}
    \item{$\varepsilon= 0$, if $P_Y \in \Gamma(P_X, \lambda, L_{max})$;}
    \item{${\displaystyle \varepsilon \ne 0}$, if $P_Y \notin \Gamma(P_X, \lambda, L_{max})$}.
\end{enumerate}
\label{theo.payoff_KS_as}
\end{theorem}
Given two pmf's $P_X$ and $P_Y$, a maximum distortion $L_{max}$ and the desired false positive error exponent $\lambda$, Theorem \ref{theo.payoff_KS_as} permits to understand whether D may ever succeed to make the false negative error probability vanishingly small and thus {\em win} the game. Then, $\Gamma (P_X, \lambda, L_{max})$ can be interpreted as the region with the sources that cannot be reliably distinguished from $P_X$ guaranteeing a false positive error exponent at least equal to $\lambda$ in the presence of an adversary with allowed distortion $L_{max}$, where by {\em reliably distinguished} we mean distinguished in such a way to grant a strictly positive error exponent for $P_{fn}$. A geometric interpretation of Theorem \ref{theorem2} is given in Figure \ref{fig.theo2}.

\begin{figure}[t!]
\centering \includegraphics[width = 0.65\columnwidth]{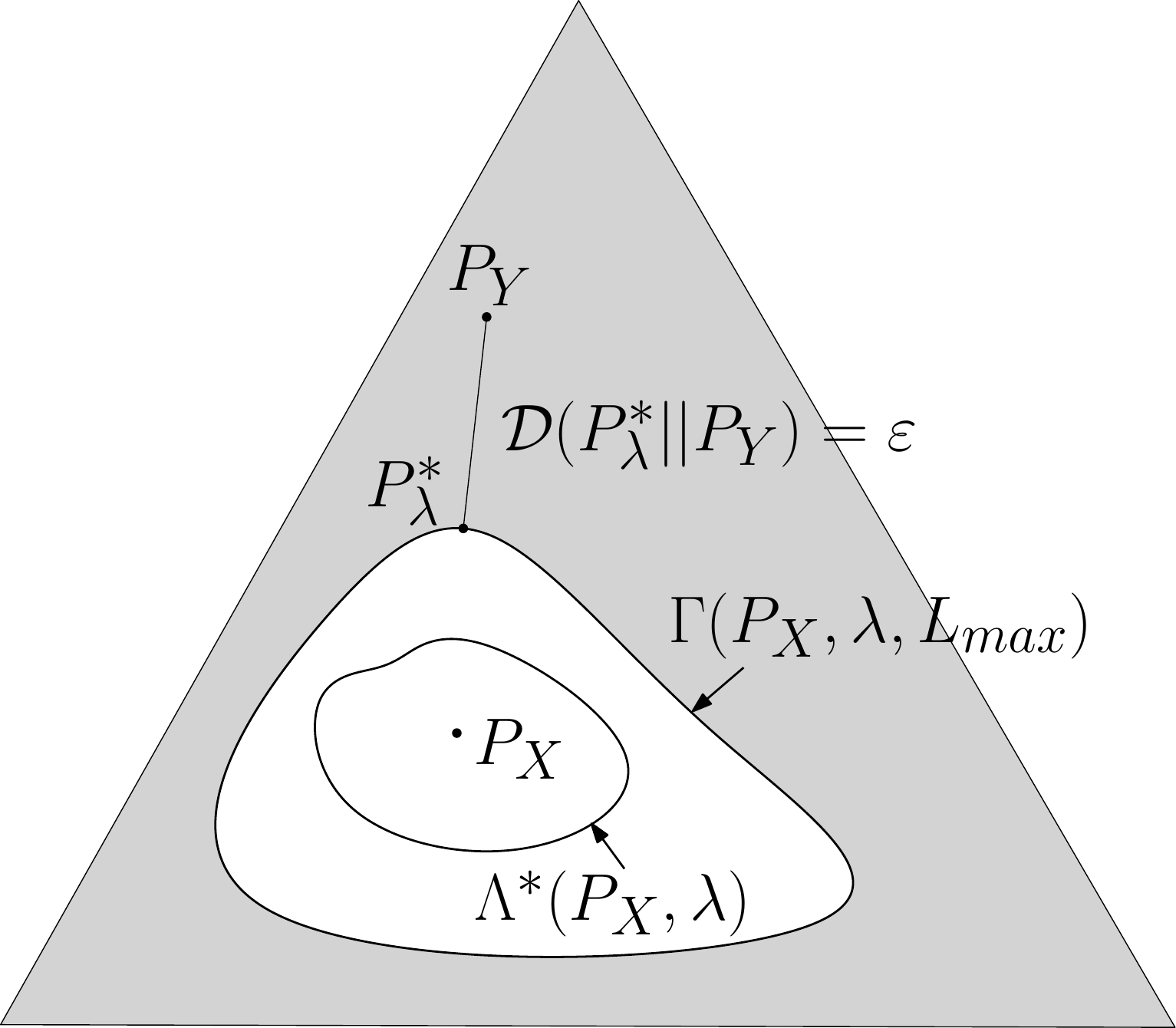}
\caption{Geometric interpretation of $\Gamma(P_X, \lambda, L_{max})$ and $\Lambda^*(P_X,\lambda)$ by the light of Theorem \ref{theorem2}. }
\label{fig.theo2}
\end{figure}

\section{The security margin}
\label{sec.SM}

In this section, we use the optimal transport interpretation of the attacker's strategy to introduce a measure of source distinguishability in the set-up defined by the $SI_{ks}$ game.

\subsection{Characterization of the indistinguishability region using Optimal Transportation}
\label{sec.SM_opt}

To start with, we find it convenient to rephrase the results described in the previous section as an optimal transport problem \cite{Vill09}.

Let $P$ and $Q$ be two pmf's defined over the same finite alphabet, and let $c(i,j)$ be the cost of transporting the $i$-th symbol into the $j$-th one. In one of its instances, optimal transport theory looks for the transportation map
that transforms $P$ into $Q$ by minimizing the average cost of the transport. By using the notation introduced in the previous section, this corresponds to solving the following minimization problem:
\begin{equation}
\label{eq.transport_problem}
   \min_{S_{YZ} : S_{Y} = P, S_{Z} = Q} ~ \sum_{i,j} S_{YZ}(i,j) c(i,j).
\end{equation}
A nice interpretation of the problem defined by equation (\ref{eq.transport_problem}) is obtained by interpreting the pmf's $P$ and $Q$ as two different ways of piling up a certain amount of earth, and $c(i,j)$ as the cost necessary to move a unitary amount of earth from position $i$ to position $j$. In this case, the minimum cost achieved in (\ref{eq.transport_problem}) can be seen as the minimum effort required to turn one pile into the other. Due to such a viewpoint, in computer vision applications, the minimum in equation (\ref{eq.transport_problem}) is usually known as Earth Mover Distance ({\em EMD}) between $P$ and $Q$, \cite{RTG00}. However, while the definition of the {\em EMD} given in \cite{RTG00} refers in general to signatures (non-normalized distributions with unequal masses), here the pilings of earth $P$ and $Q$ are probability mass functions. In this case, when $c(i,j) = d(i,j)^p$ for some distance measure $d$ (with $p \ge 1$), the {\em EMD} has a more general statistical meaning. Given two random variables with probability distributions $P_X$ and $P_Y$, the {\em EMD} between $P_X$ and $P_Y$ corresponds to the minimum expected $p$-th power distance between the random variables $X$ and $Y$ taken over all joint probability distributions $P_{XY}$  with marginal distributions respectively equal to $P_X$ and $P_Y$:
%
\begin{equation}
\label{Wasserstein_dist}
EMD_{d^p}(P_X,P_Y) = \min_{P_{XY}:\tiny{\substack{\sum_y P_{XY} = P_X \\ \sum_x P_{XY} = P_Y}}} E_{XY} [d(X,Y)^p].
\end{equation}
In transport theory terminology, expression (\ref{Wasserstein_dist}) is the $p$-th power of the Wasserstein distance \cite{rachev1998mass}, \cite{Vill09} (or the Monge-Kantorovich metric of order $p$ \cite{Vill03}, \cite{rachev1985monge}).  In particular, when $c(i,j) = |i-j|^2$ (i.e. $d(i,j) = |i-j|$ and $p=2$) the earth mover distance is equivalent to the squared Mallows distance between $P_X$ and $P_Y$ \cite{EMDMallows}, that is
\begin{equation}
\label{Mallow_dist}
EMD_{L_2^2}(P_X,P_Y) = \min_{P_{XY}:\tiny{\substack{\sum_y P_{XY} = P_X \\ \sum_x P_{XY} = P_Y}}} E_{XY} [|X - Y|^2].
\end{equation}
In the following, we will continue to refer to (\ref{eq.transport_problem}) as {\em EMD(P,Q)}. We also observe that even if we introduced the {\em EMD} by considering finite-alphabet sources, there is no need to restrict the definition in \eqref{Wasserstein_dist} and \eqref{Mallow_dist} to discrete random variables. In fact, in the second part of the paper, we will extend our analysis and use the {\em EMD} to measure the distinguishability of continuous sources.

Optimal transport theory permits us to rewrite the indistinguishability region in a more compact and easier-to-interpret way. In fact, it is immediate to see that equation (\ref{eq.indistinguib_limit}) can be rewritten as:
\begin{align}
\label{eq.indistinguib_EMD}
    \Gamma(P_X, \lambda, L_{max}) & =  \\ \nonumber
    \{P \in \PP : & ~\exists ~ Q \in {\Lambda^*}(P_X, \lambda) ~ \text{s.t} ~ \text{{\em EMD}}(P,Q) \le L_{max}  \},
\end{align}
where in the definition of the {\em EMD} $c(i,j)$ corresponds to the distortion metric  used to constraint the strategies available to the attacker.


\subsection{Security Margin definition}
\label{subsec.SMdef}

We now study the behavior of $\Gamma(P_X, \lambda, L_{max})$ when $\lambda \rarrow 0$. Doing so will allow us to investigate whether two sources $X$ and $Y$ are ultimately distinguishable in the setting defined by the $SI_{ks}$ game.

The rationale behind our analysis derives directly from equations (\ref{eq.indistinguib_limit}) and (\ref{eq.lambda_limit}). In fact, it is easy to see that decreasing $\lambda$ in the definition of $\SS_D$ leads to a more favorable game for the defender, since he can adopt a smaller acceptance region and obtain a larger payoff. Stated in another way, from D's perspective, evaluating the behavior of the game for $\lambda \rarrow 0$ corresponds to exploring the best achievable false negative error exponent, when $P_{fp}$ tends to 0 exponentially fast.

More formally, we start by proving the following property.
\begin{property}
\label{prop.inclusion}
For any two values $\lambda_1$ and $\lambda_2$ such that $\lambda_2 < \lambda_1$, $\Gamma(P_X, \lambda_2, L_{max}) \subseteq \Gamma(P_X, \lambda_1, L_{max})$.
\end{property}
\begin{proof}
The property follows immediately from equation (\ref{eq.indistinguib_EMD}) by observing that $\Gamma(P_X, \lambda, L_{max})$ depends on $\lambda$  only through the acceptance region ${\Lambda^*}(P_X, \lambda)$, for which we obviously have ${\Lambda^*}(P_X, \lambda_2) \subseteq {\Lambda^*}(P_X, \lambda_1)$ whenever $\lambda_2 < \lambda_1$.
\end{proof}
Thanks to Property \ref{prop.inclusion}, we can compute the limit of the false negative error exponent when $\lambda$ tends to zero, as summarized in the following theorem (somewhat resembling Stein's Lemma \cite{CandT}).
\begin{theorem}
\label{theorem_EMD}
Given two sources $X \sim P_X$ and $Y \sim P_Y$ and a maximum average per-letter distortion $L_{max}$ (defined according to an additive distortion measure), let us adopt the following definition:
\begin{equation}
\Gamma(P_X,L_{max}) = \{P \in \PP: \text{\em EMD}(P,P_X) \le L_{max}\};
\label{Gamma_X}
\end{equation}
then the maximum achievable false negative error exponent $\varepsilon$ for the $SI_{ks}$ game is
\begin{equation}
\lim_{\lambda \rarrow 0} \lim_{n \rarrow \infty} - \frac{1}{n} \log P_{fn} = \min_{P \in \Gamma(P_X,L_{max})} \DD(P || P_Y).
\label{best_e_e}
\end{equation}
\end{theorem}
\begin{proof}
The innermost limit in (\ref{best_e_e}) defines the error exponent for a fixed $\lambda$, say it $\varepsilon(\lambda)$. Thanks to equation (\ref{eq.asymptotic_theorem}), we know that
\begin{equation}
\label{eq.epslambda}
    \lim_{n \rarrow \infty} - \frac{1}{n} \log P_{fn} = \varepsilon(\lambda) = \min_{P \in \Gamma(P_X, \lambda, L_{max})} \DD(P||P_Y).
\end{equation}
Then, according to  Property \ref{prop.inclusion}, the sequence $\varepsilon(\lambda)$ is monotonically non decreasing as $\lambda$ decreases. In addition, since $\Gamma(P_X,L_{max}) \subseteq \Gamma(P_X, \lambda, L_{max})$ $\forall \lambda$, for any  $\lambda > 0$, we have:
\begin{equation}
\label{e.e.lambda_upper}
\varepsilon(\lambda) \le  \min_{P \in \Gamma(P_X,L_{max})} \DD(P || P_Y).
\end{equation}
Being $\varepsilon(\lambda)$ bounded from above and non-decreasing, the limit for $\lambda \rarrow 0$ exists and is finite. We must now prove that the limit is indeed equal to $\min_{P \in \Gamma(P_X,L_{max})} \DD(P || P_Y)$.
Let $P^*_0$ be the point achieving the minimum in (\ref{best_e_e}) and $P^*_{\lambda}$ the point achieving the minimum on the set $\Gamma(P_X, \lambda, L_{max})$, i.e. the point achieving the minimum in equation (\ref{eq.asymptotic_theorem}) (see Figure \ref{fig.theo2} for a pictorial representation of $P^*_{\lambda}$). Due to Lemma \ref{lemma_A} (Appendix \ref{app.lemma_A}),
for any arbitrarily small $\tau$, we can choose a small enough $\lambda$ such that, for any $P$ in $\Gamma(P_X, \lambda, L_{max})$, a pmf $P'$ in $\Gamma(P_X,L_{max})$ exists whose distance from $P$ is lower than $\tau$. By taking $P = P^*_{\lambda}$ and exploiting the continuity of the $\DD$ function, we have
\begin{equation}
\DD(P'||P_Y) \le \min_{P \in \Gamma(P_X, \lambda, L_{max})} \DD(P || P_Y) + \delta(\tau),
\label{relation_lambda_e_0}
\end{equation}
for some $P' \in \Gamma(P_X,L_{max})$ and some value $\delta(\tau)$ such that $\delta(\tau) \rarrow 0$ as $\tau \rarrow 0$. A fortiori, relation (\ref{relation_lambda_e_0}) holds for $P' = P^*_0$ and then we can write
\begin{align}
\label{divergence_continuity_Gamma}
\varepsilon(\lambda) & = \min_{P \in \Gamma(P_X, \lambda, L_{max})} \DD(P || P_Y) \\ \nonumber
& \ge \min_{P \in \Gamma(P_X,L_{max})} \DD(P || P_Y) - \delta(\tau).
\end{align}
where $\delta(\tau)$ can be made arbitrarily small by decreasing $\lambda$.
Equation (\ref{divergence_continuity_Gamma}), together with equation (\ref{e.e.lambda_upper}), shows that we can get arbitrarily close to $\min_{P \in \Gamma(P_X,L_{max})} \DD(P || P_Y)$, by making $\lambda$ small enough, hence proving that $\min_{P \in \Gamma(P_X,L_{max})} \DD(P || P_Y)$ is the limit of the sequence $\varepsilon(\lambda)$ as $\lambda \rarrow 0$.
\end{proof}

Figure \ref{fig.theo3} gives a geometric interpretation of Theorem \ref{theorem_EMD}. The figure is obtained from Figure \ref{fig.theo2} by observing that when $\lambda \rarrow 0$ the optimum acceptance region collapses into the single pmf $P_X$, i.e., $\Lambda^* = \{P_X\}$.

\begin{figure}[t!]
\centering \includegraphics[width = 0.65\columnwidth]{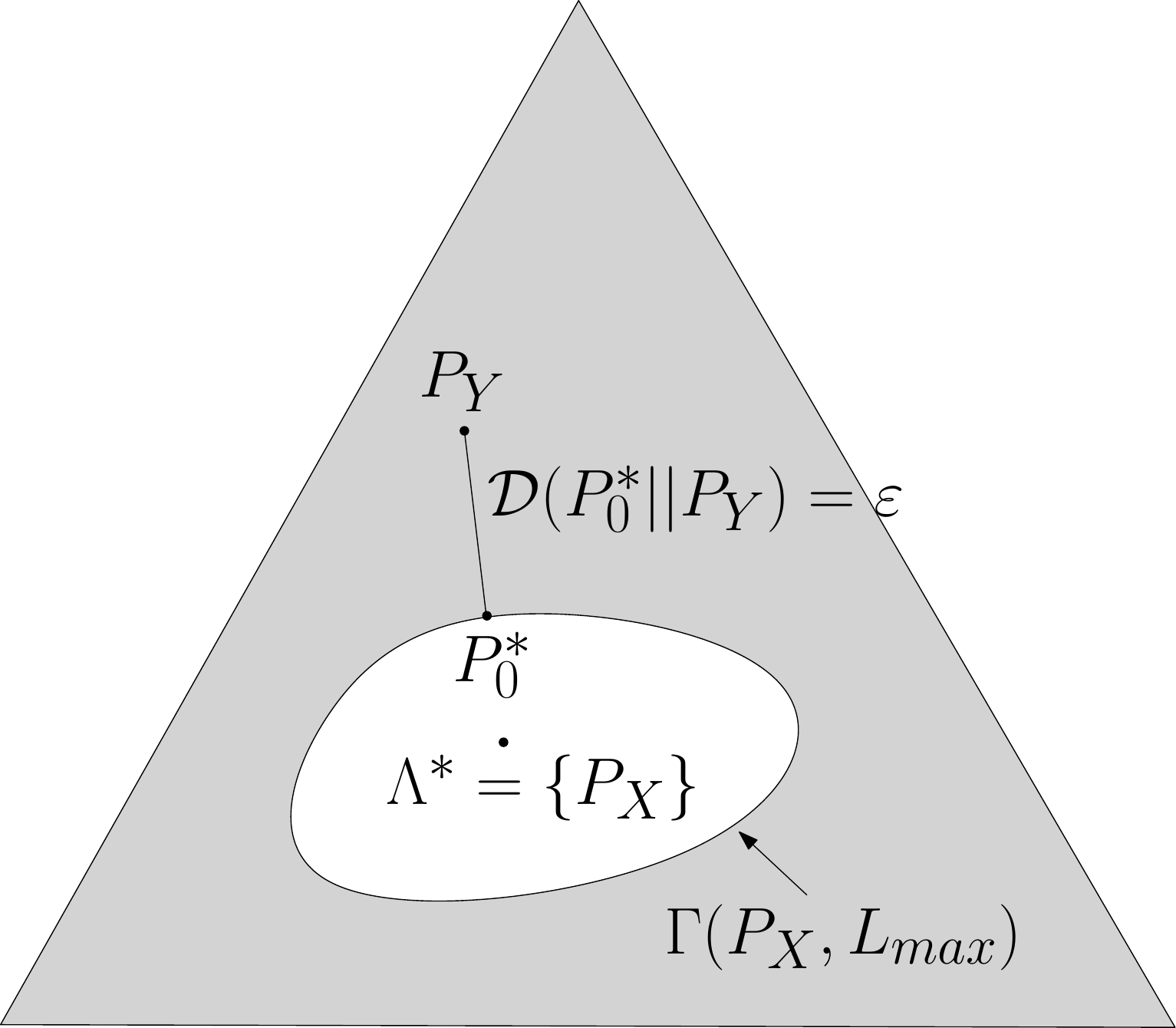}
\caption{Geometric interpretation of $\Gamma(P_X, L_{max})$ and $P^*_0$ by the light of Theorem \ref{theorem_EMD}. }
\label{fig.theo3}
\end{figure}

By the light of Theorem \ref{theorem_EMD}, $\Gamma(P_X, L_{max})$ is the smallest indistinguishability region for the $SI_{ks}$ game. Moreover, from equation (\ref{Gamma_X}), we see that the distinguishability of two pmf's (in the $SI_{ks}$ setting) ultimately depends on their {\em EMD}. In fact, if {\em EMD}$(P_Y,P_X) > L_{max}$, the defender is able to distinguish $X$ from $Y$ by adopting a sufficiently small $\lambda$. On the contrary, if {\em EMD}$(P_Y,P_X) \le L_{max}$, there is no positive value of $\lambda$ for which the sequences emitted by the two sources can be asymptotically distinguished.

By adopting a different perspective, given two sources $X$ and $Y$, one may ask which is the maximum attacking distortion for which D can distinguish $X$ and $Y$ despite the presence of the adversary. The answer to this question follows immediately from Theorem \ref{theorem_EMD} and leads naturally to the following definition.
\begin{definition}[Security Margin]
\label{def.SMdef}
Let $X \sim P_X$ and $Y \sim P_Y$ be two discrete memoryless sources. The maximum average per-letter distortion for which the two sources can be reliably distinguished in the $SI_{ks}$ setting is called Security Margin and is given by
\begin{equation}
\SS \MM(P_Y, P_X) = \text{\em EMD}(P_Y,P_X).
\label{security_margin_discrete}
\end{equation}
\end{definition}
\noindent Interestingly, the {\em EMD} is a symmetric function of $P_X$ and $P_Y$ \cite{RTG00}, and hence the security margin does not depend on the role of $X$ and $Y$ in the test, i.e. $\SS\MM(P_X, P_Y) = \SS\MM(P_Y,P_X)$. The security margin is a powerful measure summarizing in a single quantity how securely two sources can be distinguished (in the $SI_{ks}$ setting).

It is worth remarking that the security margin between two sources pertains to the {\em security} of the hypothesis test behind the source identification problem and not to its {\em robustness}, since it is measured at the equilibrium of the game, i.e. by assuming that both the players of the game make optimal choices. To better exemplify the above concept, let us consider the simple case of two binary sources. Specifically, let $X$ and $Y$ be two Bernoulli sources with parameters $p = P_X(1)$ and $q=P_Y(1)$ respectively. Let also assume that the distortion constraint is expressed in terms of the Hamming distance between the sequences, that is $d(i,j) = 0$ when $i=j$ and $1$ otherwise. Without loss of generality let $p > q$. The distortion associated to a transportation map $S_{XY}$ can be written as:
\begin{equation}
\sum_{i,j} S_{YX}(i,j) d(i,j) = S_{YX}(0,1) + S_{YX}(1,0).
\end{equation}
Since $p > q$, it is easy to conclude that the minimum of the above expression is obtained when $S_{YX}(1,0) = 0$ (intuitively, if the source $X$ outputs more 1's than $Y$, it does not make any sense to turn the 1's emitted by $Y$ into 0's). As a consequence, to satisfy the constraint $S_X(1) = p$ we must let $ S_{YX}(0,1) = p - q$, yielding $\SS \MM(P_Y,P_X) = p-q$, or more generally $|p-q|$. We can conclude that if the attacker is allowed to introduce an average Hamming distortion larger or equal than $|p-q|$, then there is no way for the defender to distinguish between the two sources. This is not the case if the output of the source $Y$ passes through a binary symmetric channel with crossover probability equal to $|p-q|$, since the output of the channel will still be distinguishable from the sequences emitted by $X$. Consider, for example, a simple case in which $q = 1/2$ and $p >1/2$. Regardless of the crossover probability, the output of the channel will still be a binary source with equiprobable symbols, which is distinguishable from $X$ given that $p > 1/2$. In other words, in the set up defined by the $SI_{ks}$ game, the two sources can not be distinguished securely in the presence of an attacker introducing a distortion equal to $|p-q|$, while they can be distinguished even if the output of the source $Y$ passes through a noisy channel introducing the same average distortion introduced by the attacker.

\section{Extension to source identification with training data}
\label{sec.SI_tr}

In this section we extend the previous analysis to the case of source identification with training data ($SI_{tr}$), in order to provide a measure of source distinguishability, in the more general setup studed in \cite{BTtit}. In such a scenario, the two sources $X$ and $Y$ are not completely known to D and A, so they must base their actions on the knowledge of a training sequence drawn from $X$ (the source under the null hypothesis). This is a very interesting scenario bringing the analysis closer to real applications, in which a precise statistical model of the to-be-distinguished sources is usually not available. In \cite{BTtit}, it is proven that the source identification game with training data is more favorable to the attacker than the $SI_{ks}$ game. Then one could argue that in the $SI_{tr}$ setup the security margin between the two sources is smaller, implying that a lower distortion is sufficient to the attacker to make the sources undistinguishable. The remarkable result that we will prove in this section is that this is not the case, hence showing that the ultimate distinguishability of two sources is the same for the two games.

\subsection{The source identification game with training data ($SI_{tr}$)}

In order to present our analysis in a self-contained way, in this section we summarize the main results proven in \cite{BTtit}. Once again, we will do so by adopting a transportation theory perspective for the definition of the attacker's optimum strategy.

Let us start by giving a rigorous definition of the source identification game with training data.

{\em Defender's strategies.}  In the $SI_{tr}$ game the defender must decide whether a test sequence $x^n$ has been generated by a source $X$ with unknown pmf by relying on the knowledge of an $N$-sample training sequence $t^N_D$ drawn from $X$. This is equivalent to deciding whether to accept or not the hypothesis  $H_0$ that the test and the training sequences have been generated by the same source. In this framework, the acceptance region $\Lambda$ is defined as the set with all the pairs of sequences $(x^n,t^N_D)$ that D classifies as being generated by the same source. Once again, we limit the action of D to a first order analysis of $x^n$ and $t_D^N$. This is equivalent to require that the acceptance region for hypothesis $H_0$ is a union of pairs of type classes, or equivalently, pairs of types $(P,Q)$, where $P \in \mathcal{P}_n$ and $Q \in \mathcal{P}_N$.
As for the $SI_{ks}$ case, the defender must ensure that the asymptotic false positive error probability tends to zero exponentially fast at least with a certain decay rate, however since $P_X$ is not known, the constraint must be satisfied in a worst case sense, i.e. for all possible choices of $P_X$. More specifically, the space of strategies of D is given by:
\begin{equation}
    \SS_{D} = \{ \Lambda^n_{tr} \subset \mathcal{P}_n\times\mathcal{P}_N: \max_{P_X \in \mathcal{P}} P_{fp} \le 2^{-\lambda n}\},
\label{eq.SD_TR_as}
\end{equation}
where $\mathcal{P}$ is the class of discrete memoryless sources.

{\em Attacker's strategies.} Given a sequence $y^n$ drawn from a source $Y \ne X$, the goal of A is to transform $y^n$ into a sequence $z^n$ belonging to the acceptance region chosen by D while respecting a distortion constraint. Likewise the defender, all the information that the attacker has about $X$ is a $K$-long training sequence $t_A^K$. By using the same transportation theoretic formalism used in the previous section, the set of strategies of the attacker consists of all the possible ways of choosing an admissible transportation map to transform $y^n$ into $z^n$.
\begin{equation}
    \SS_{A} = \{ S^n_{YZ}(i,j;y^n,t_A^K) : S^n_{YZ}(i,j)  \in \A^n(L_{max}, P_{y^n}) \},
\label{eq.SA_TR_as}
\end{equation}
where we have explicitly indicated that the choice of the transportation map now depends also on $t_A^K$, and where the set of admissible maps is defined as in the $SI_{ks}$ case.

Depending on the relationship between $t_A^K$ and $t_D^N$, several versions of the $SI_{tr}$ game can be defined. Here we focus on the simplest case of equal training sequences, i.e. we assume $K=N$ and $t_A^N = t_D^N \triangleq t^N$. We will see later on that our analysis can be easily extended so to cover the other cases addressed in \cite{BTtit}. In addition, we force $N$ to be a linear function of $n$ with some proportionality constant $c$, i.e. $N=cn$. As discussed in \cite{BTtit}, this is the most significant case to study.

{\em The payoff.} Adopting again the Neyman-Pearson approach, the payoff corresponds to the false negative error probability, that is:
\begin{equation}
\label{equation_tr}
    u_D = -u_A = - \hspace{-0.5cm} \underset{\underset{(S^n_{Z}(j;y^n,t^N),t^N) \in \Lambda_{tr}^n}{(y^n,t^N) \in \XX^n \times \XX^{N}:}}{\sum} P_Y(y^n)P_X(t^N),
\end{equation}
where $P_X(t^N)$ is the probability that the source $X$ outputs the sequence $t^N$ and $\Lambda_{tr}^n$ is the acceptance region of the test.

{\em Equilibrium point.}
The derivation of the optimum strategy for D passes through the definition of the generalized log-likelihood ratio function $h(P_{x^n}, P_{t^N})$ defined as (\cite{Gut89,Kendall,BTtit}):
\begin{equation}
    h(P_{x^n}, P_{t^{N}}) = \DD(P_{x^n} || P_{r^{n+N}}) + c \DD(P_{t^N} || P_{r^{n+N}}),
\label{eq.h}
\end{equation}
where $P_{r^{n+N}}$ indicates the empirical pmf of the sequence $r^{n+N}$, obtained by concatenating $x^n$ and $t^N$.
The main result of \cite{BTtit} is summarized by the following theorem.
\begin{theorem}
Let
\begin{equation}
    \Lambda_{tr}^{n,*} = \left\{(P, Q) \in \PP_n \times \PP_N: h(P, Q) < \lambda - \kappa(n,c) \right\},
\label{eq.opt_lambda_Tr_as}
\end{equation}
\begin{equation}
    S_{YZ}^{n,*}(i,j;y^n,t^N) = \arg\min_{S^n_{YZ} \in \A^n(D_{max},P_{y^n})} h(S^n_{Z}, P_{t^N}).
\label{eq.opt_AD_Tr_as}
\end{equation}
where $\kappa(n,c) = |\XX| \frac{\log(n+1)(N+1)}{n}$.
Then $\Lambda_{tr}^{n,*}$ is a dominant equilibrium for D and the profile $(\Lambda_{tr}^{n,*},  S_{YZ}^{n,*}(i,j;y^n,t^N))$ is the only rationalizable equilibrium of the $SI_{tr}$ game with equal training sequences, which, then, is a dominance solvable game \cite{ChenGames}.
\label{theo.NASH_TR_as}
\end{theorem}
As for the $SI_{tr}$ game, by letting $n$ tend to infinity and by exploiting the density of rational numbers in the real line, we can study the asymptotic distinguishability of sequences emitted by any two sources. To express the final result of the above procedure, we need to introduce some definitions. First of all we need to extend the $h$ function so to make it work on general pmf's. We let:
\begin{align}
\label{eq.hc}
    h_c(P,Q) & = \DD(P||U) + c \DD(Q||U); \\ \nonumber
    U  & = \frac{1}{1+c}P + \frac{c}{1+c}Q,
\end{align}
which permits us to define the following sets:
\begin{equation}
\label{eq.Lambda_tr_inf}
    \Lambda_{tr}^*(Q,\lambda) = \{ P \in \PP : h_c(P,Q) \le \lambda\},
\end{equation}
and
\begin{align}
\label{eq.Gamma_tr_inf}
    \Gamma_{tr}(Q,\lambda,L_{max}) = \{ P \in \PP & : \exists R \in \Lambda_{tr}^*(Q,\lambda) \\ \nonumber
    & \text{ s.t. } \text{\em EMD}(P,R) \le L_{max}\}.
\end{align}
The following theorem, proved in \cite{BTtit}, states that the indistinguishability region of the $SI_{tr}$ game is given by $\Gamma_{tr}(P_X, \lambda, L_{max})$, where $P_X$ is the true distribution of the source $X$.
\begin{theorem}
\label{theorem2_tr}
For the $SI_{tr}$ game with equal training sequences available to the players, the error exponent of the false negative error probability at the equilibrium is given by:
\begin{equation}
\label{eq.fnerr_exp_TRc}
    \varepsilon_{tr}(\lambda) = \min_{R} \left[ c \cdot \DD(R || P_X) + \min_{P \in \Gamma_{tr}(R, \lambda, L_{max}) } \DD (P || P_Y)\right]
\end{equation}
leading to the following cases:
\begin{enumerate}
    \item{$\varepsilon_{tr}(\lambda)= 0$, if $P_Y \in \Gamma_{tr}(P_X, \lambda, L_{max})$;}
    \item{${\displaystyle \varepsilon_{tr}(\lambda) \ne 0}$, if $P_Y \notin \Gamma_{tr}(P_X, \lambda, L_{max})$}.
\end{enumerate}
\label{theo.payoff_TR_as}
\end{theorem}

From the above theorem we see that the sources that cannot be asymptotically distinguished from $P_X$ are those inside $\Gamma_{tr}(P_X, \lambda, L_{max})$. The geometrical interpretation is similar to the one given in Figure \ref{fig.theo2} for Theorem \ref{theo.payoff_KS_as} where now the acceptance region is given by $\Lambda_{tr}^*(P_X,\lambda)$ and the indistinguishability region is $\Gamma_{tr}(P_X, \lambda, L_{max})$.

We point out that the only difference with respect to the case of known sources consists in the asymptotic acceptance region ${\Lambda_{tr}^*}(P_X, \lambda)$, which is strictly larger than ${\Lambda^*}(P_X, \lambda)$, given that $h_c$ function is always lower than $\DD$ (see \cite{BTtit} for the proof). As a consequence, it is straightforward to argue that $\Gamma_{tr}(P_X, \lambda, L_{max}) \supset \Gamma(P_X, \lambda, L_{max})$.

\subsection{Security margin for the $SI_{tr}$ game}
\label{subsec.SMTR}
We now study the behavior of the $SI_{tr}$ game when $\lambda \rarrow 0$ so to investigate the best achievable performance for the defender in the case of training-based decision. To start with, we observe that the divergence and the $h_c$ function share a similar behavior, in that they are convex functions and both $\DD(P||Q)$ and $h_c(P,Q)$ are equal to zero if and only if $P=Q$. This permits to extend Property \ref{prop.inclusion} to the set $\Gamma_{tr}$ yielding:
\begin{property}
\label{prop.inclusion_TR}
For any two values $\lambda_1$ and $\lambda_2$ such that $\lambda_2 < \lambda_1$, $\Gamma_{tr}(P_X, \lambda_2, L_{max}) \subseteq \Gamma_{tr}(P_X, \lambda_1, L_{max})$.
\end{property}
In a similar way, Lemma \ref{lemma_A} can be extended to the set $\Gamma_{tr}(R, \lambda, L_{max})$ (Appendix \ref{app.lemma_A}).

We are now ready to prove the counterpart of Theorem \ref{theorem_EMD} for the $SI_{tr}$ game.
\begin{theorem}
\label{theorem_EMD_tr}
Given two sources $X \sim P_X$ and $Y \sim P_Y$ and a maximum allowable average per-letter distortion $L_{max}$ (defined according to an additive distortion measure),
the maximum achievable false negative error exponent for the $SI_{tr}$ game is
\begin{equation}
\label{best_e_e_tr}
\lim_{\lambda \rarrow 0}\varepsilon_{tr}(\lambda) = \min_{R} \big[ c \cdot \DD(R || P_X) + \min_{P \in \Gamma(R, L_{max}) } \DD (P || P_Y)\big],
\end{equation}
where $\Gamma(R, L_{max})$ is defined as in (\ref{Gamma_X}) by replacing $P_X$ with $R$\footnote{Note that when $\lambda$ tends to 0, we do not need anymore to differentiate between the $SI_{ks}$ and $SI_{tr}$ games in the definition of $\Gamma(R, L_{max})$.}.
\end{theorem}
\begin{proof}
The proof goes along the same line of the proof of Theorem \ref{theorem_EMD}.
From Property \ref{prop.inclusion_TR}, we see immediately that $\varepsilon(\lambda)$ is non-increasing when $\lambda$ decreases, since the innermost minimization in equation (\ref{eq.fnerr_exp_TRc}) is taken over a smaller set when $\lambda$ decreases. Then, by the same token, we have:
\begin{equation}
\varepsilon_{tr}(\lambda) \le \min_{R} \big( c \DD(R||P_X) + \min_{P \in \Gamma(R,D_{max})} \DD(P || P_Y)\big).
\end{equation}
This implies that $\lim_{\lambda \rarrow 0} \varepsilon(\lambda)$ exists and is finite. Given that Lemma \ref{lemma_A} still holds for the set $\Gamma_{tr}(R,\lambda,L_{max}) ~ \forall R$, we can reason as in the proof of Theorem \ref{theorem_EMD} to conclude that:
\begin{equation}
\min_{P \in \Gamma_{tr}(R, \lambda, L_{max})} \DD(P || P_Y) \ge \min_{P \in \Gamma(R,L_{max})} \DD(P || P_Y) - \delta(\tau),
\label{divergence_continuity_Gamma_tr}
\end{equation}
where $\delta(\tau)$ can be made arbitrarily small by decreasing $\lambda$. By adding the term $c \DD (R|| P_X)$ to both sides of (\ref{divergence_continuity_Gamma_tr})  we obtain:
\begin{align}
\label{low_bound_proof_tr_2}
c \DD(R||P_X)& + \min_{P \in \Gamma_{tr}(R, \lambda, L_{max})} \DD(P || P_Y)  \ge \\
& c \DD(R||P_X) +  \min_{P \in \Gamma(R,L_{max})} \DD(P || P_Y) - \delta(\tau).\nonumber
\end{align}
Given that \eqref{low_bound_proof_tr_2} holds for any $R \in \PP$, we can write:
\begin{align}
\varepsilon_{tr}(\lambda) & = \min_R \big[ c \DD(R||P_X) + \min_{P \in \Gamma_{tr}(R, \lambda, L_{max})} \DD(P || P_Y) \big] \\ \nonumber
& \ge \min_{R} \big[ c \DD(R||P_X) + \min_{P \in \Gamma(R,L_{max})} \DD(P || P_Y)\big] - \delta(\tau),
\end{align}
which concludes the proof due to the arbitrariness of $\delta(\tau)$.
\end{proof}
A consequence of Theorem \ref{theorem_EMD_tr} is that $\lim_{\lambda \rarrow 0} \varepsilon(\lambda) = 0$ if and only if $P_Y \in \Gamma(P_X,L_{max}$), which then can be seen as the smallest indistinguishability region for the $SI_{tr}$ game. We conclude that the smallest indistinguishability regions for the two cases are the same thus implying that the security margin for the $SI_{tr}$ setting, say $\SS\MM_{tr}$, is the same of the $SI_{ks}$ game, that is
\begin{equation}
\SS\MM_{tr}(P_X,P_Y) = \text{\em EMD}(P_X,P_Y).
\end{equation}
We remark that, for any allowed distortion $L_{max} < \text{\em EMD}(P_X,P_Y)$, the minimum value of the false positive error exponent ($\lambda$) which allows the defender to take a reliable decision in the $SI_{tr}$ setting is lower than that in the $SI_{ks}$ setting. However, the difference between the two settings regards the decay rate of the error probabilities, not the ultimate distinguishability of the sources.

We conclude this section with a brief discussion on the $SI_{tr}$ game with different training sequences ($t_{D}^{N} \ne t_{A}^K$). Such a scenario provides a more realistic model in which the attacker is not able to compute exactly the acceptance region adopted by the Defender. It is known from \cite{BTtit} that, as long as the length of both sequences grows linearly with $n$, the indistinguishability region is equal to that of the game with equal training sequences. By relying on this result, it is not difficult to prove that the security margin remains the same even for such version of the game.

\section{Security margin computation}
\label{sec.SMcomputation}

In this section we address the problem of the actual computation of the security margin for two generic sources. By following the analysis given so far, we focus on the case of discrete sources, however at the end of the section we extend the analysis so to cover continuous sources as well.

Given two discrete sources $X \sim P_X$ and $Y \sim P_Y$, the computation of the security margin requires the evaluation of {\em EMD$(P_X, P_Y)$}. A closed form solution can be found only in some simple cases (see Section \ref{sec.example1} and \ref{sec.example2}).
More generally, the {\em EMD} between two sources can be computed by resorting to numerical analysis, and in fact, due to its wide use as a similarity measure in computer vision applications, several efficient algorithms have been proposed (see \cite{PeleEMD} for example). In the following, we describe a fast iterative algorithm for the computation of the {\em EMD} between any two sources assuming that the distortion (or cost) function has the general form:
\begin{equation}
\label{eq.dist_EMD}
    d(i,j) = |i-j|^p,
\end{equation}
with $p \ge 1$. This is a case of great interest for $p=1$ and $p=2$, according to which the distortion between $y^n$ and the attacked sequence $z^n$ corresponds, respectively, to the $L_1$ and $L_2^2$ distance.

\subsection{Hoffman's greedy algorithm for computing $\SS\MM$}
\label{sec.Huffman_NWC}

Let us assume that $X$ and $Y$ are discrete sources with alphabets $\XX$ and $\YY$. The transportation problem we have to solve for computing  $\SS \MM(P_Y,P_X)$, i.e. {\em EMD}$(P_Y,P_X)$, is known in modern literature as {\em Hitchcock transportation problem} \cite{Hitchcock}\footnote{This is the discrete version of the Monge-Kantorovich mass transportation problem \cite{rachev1998mass}.}, which, in turn, can be formulated as a linear programming problem in the following way:
%
\begin{equation}
\label{linear_programming1}
\text{\em EMD}(P_X,P_Y) = \min_{S_{XY}} \sum_{i,j}  d(i,j) S_{XY}(i,j),
\end{equation}
where $S_{XY}$ must satisfy the linear constraints:
\begin{align}\label{linear_programming2}
\sum_{j} S_{XY}(i,j) & = P_X(i) & \forall i \in \XX\nonumber\\
\sum_{i} S_{XY}(i,j) & = P_Y(j) & \forall j \in \YY\nonumber\\
S_{XY}(i,j) & \ge 0  \quad \forall i,j,
\end{align}
and where, by referring to the original Monge formulation\footnote{Monge is considered the founding father of optimal transport \cite{monge1781}.}, $S_{XY}(i,j)$ denotes the quantity of soil shipped from location (source) $i$ to location (sink) $j$ and $d(i, j)$ is the cost for shipping a unitary amount of soil from $i$ to $j$.

A Transportation Problem (TP) like the one defined by equations (\ref{linear_programming1}) and (\ref{linear_programming2}) is a particular minimum cost flow problem \cite{Network} which, being linear, can be solved through the simplex method \cite{OR}.
In general, the solution of TP depends on the cost function $d(\cdot,\cdot)$, however there are some classes of cost functions for which the solution can be found through a simple greedy algorithm. Specifically, the algorithm proposed by A.J. Hoffman in 1963 \cite{hoffman1963}, allows to solve the transportation problem whenever $d(\cdot,\cdot)$ satisfies the so called Monge property \cite{Monge}, that is when:
\begin{equation}
d(i,j) + d(r,s) \le d(i,s) + d(r,j),
\end{equation}
$\forall (i,j,r,s)$ such that $1 \le i < r \le |\XX|$ and  $1 \le j < s \le |\YY|$. \\
It is easy to verify that Monge property is satisfied by any cost function of the form in (\ref{eq.dist_EMD}), and, more in general, by any convex function of the quantity $|i-j|$.
The iterative procedure proposed by Hoffman to solve the optimal transport problem is known as {\em north-west corner ({\em NWC}) rule} \cite{hoffman1963} and can be described as follows. Take the bin of $\XX$ with the smallest value and start moving its elements into the bin with the smallest value in $\YY$. When the smallest bin of $\YY$ is filled, go on with the second smallest bin in $\YY$. Similarly, when the smallest bin in $\XX$ is emptied, go on with the second smallest bin in $\XX$. The procedure is iterated until all the bins in $\XX$ have been moved into those of $\YY$. Let $i^{low}$ ($i^{up}$) and $j^{low}$ ($j^{up}$) denote the lower (upper) non-empty bins of $\XX$ and $\YY$ respectively. A pseudocode description of the {\em NWC} rule is given below.
\begin{enumerate}
\label{greedy_algorithm}
   \item Initialize: $i:=i^{low}$, $j:=j^{low}$.
   \item Set $S_{XY}(i,j) := \min\{P_X(i), P_Y(j)\}$.
   \item Adjust the \lq supply' distribution $P_X(i) := P_X(i) - S_{XY}(i,j)$ and the \lq demand' distribution $P_Y(j) := P_Y(j) - S_{XY}(i,j)$.\\
   If $P_X(i) = 0$ then $i := i+1$ and if $P_Y(j) = 0$ then $j := j+1$.
   \item  If $j < j^{up}$ or $P_Y(j^{up}) > 0$  go back to Step 2).
\end{enumerate}
The above procedure is described graphically in Figure \ref{fig.opt_trans}. In the figure, we chose two distributions with disjoint supports for sake of clarity, however the procedure is valid regardless of how the two distributions are spread along the real line. Interestingly, the {\em NWC} rule does not depend explicitly on the cost matrix, so the transportation map obtained through it is the same regardless of the Monge cost.
According to Hoffman's greedy algorithm, when the cost function satisfies Monge's property, the {\em EMD} can be computed in linear running time: the number of elementary operations, in fact, is at most equal to $|\XX| + |\YY|$\footnote{For sake of simplicity, the iterative algorithm described by the pseudocode spans all the bins between the minimum and the maximum non-empty bins. However, only the values $i \in \XX$ and $j \in \YY$ must be considered given that for all the empty bins $i$ and $j$ we have $S_{XY}(i,j) = 0$.}. This represents a dramatic simplification with respect to the complexity required to solve a general Hitchcock transportation problem (see for example \cite{orlin1993}).\\


As detailed below, in some cases, it is possible to derive a closed form expression for the security margin.

\subsubsection{Uniform sources with different cardinalities}
\label{sec.example1}

Let $X$ and $Y$ be two uniform pmf's with alphabets $\XX$ and $\YY$ such that $|\XX| = \alpha |\YY|$, with $\alpha \in \mathbb{N}$. In this case, thanks to Hoffman's algorithm we can express $\SS \MM (P_X, P_Y)$ in closed form:
\begin{equation}
 \SS \MM_{L_p^p}(P_X,P_Y) = \frac{1}{|\mathcal{Y}|}\sum_{i=0}^{|\XX| - 1} \sum_{j=0}^{\alpha - 1} (|i^{low} - j^{low}| - j - (\alpha - 1)i)^p,
\end{equation}
%
The formula implicitly assumes that $j^{low} > i^{low}$, the extension to the case in which such a relationship does not hold being immediate.
\begin{figure}[t!]
\centering \includegraphics[width = 0.95\columnwidth]{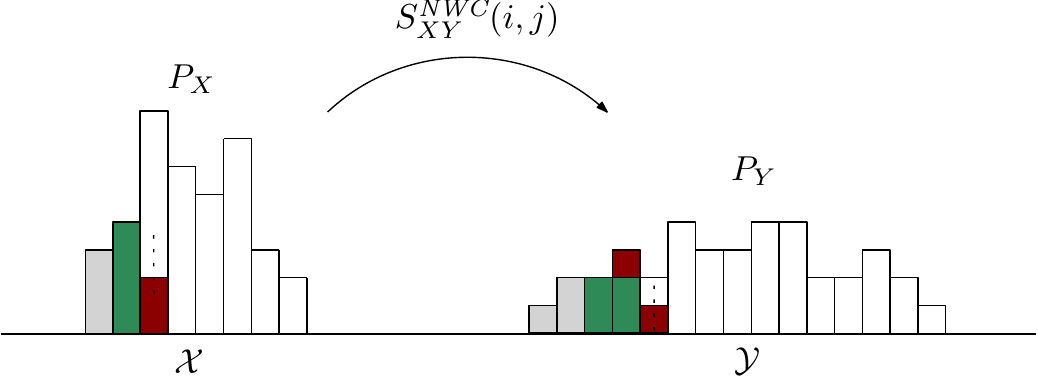}
\caption{Graphical representation of the north-west corner rule for the earth mover transportation problem (Monge problem). $P_X$ and $P_Y$ are two generic earth piles (source and sink) $\XX$ and $\YY$, while $S_{XY}^{\text{\em NWC}}(i,j)$ denotes the amount of earth moved from location $i$ to $j$.}
\label{fig.opt_trans}
\end{figure}

\subsubsection{Security Margin under the $L_1$ distance}
\label{sec.example2}

If the distortion function corresponds to the $L_1$ distance, the {\em EMD} (and hence the security margin) assumes a particularly simple form. Specifically, by applying the flow decomposition principle \cite{flowDecTP1962}, it is easy to see that the security margin between $P$ and $Q$ can be calculated as follows:
\begin{equation}
\label{SM_L1}
\SS\MM_{L_1}(P,Q) = \sum_{i= \min\{i^{low}, j^{low}\}}^{\max\{i^{up}, j^{up}\}}  \left|\sum_{s = 1}^{i} (P(s) - Q(s))\right|.
\end{equation}
%


\subsection{Continuous sources}

The analysis carried out in the previous sections is limited to discrete sources. When continuous sources are considered, we can quantize the probability density functions (pdf's) of the sources and apply the analysis for discrete sources. By letting the quantization step tend to zero, the {\em EMD} between $P_X$ and $P_Y$ can still be regarded as the security margin between the two sources. In this case, a general expression for the $\SS\MM$ can be derived by considering the {\em continuous transportation problem} (CTP), known as Monge-Kantorovic formulation of the mass transportation problem:
\begin{equation}
\label{eq.Monge_Kant}
\SS\MM(P_X,P_Y)  = \min_{S_{XY}(x,y)} \int \! \! \int c(x,y) S_{XY}(x,y)dx dy,
\end{equation}
subject to the constraints
\begin{align}
\label{eq.Monge_Kant_const}
\int S_{XY}(x,y) dx  = P_Y(y) \\ \nonumber
\int S_{XY}(x,y) dy  = P_X(x) \\ \nonumber
S_{XY}(x,y) \ge 0  \quad \forall x,y,
\end{align}
where $c$ is a continuous cost function $c(x,y): X \times Y \rightarrow \mathds{R}$. If $c(x,y)$ satisfies the continuous Monge property \cite{Monge}, that is if:
\begin{equation}
c(x,y) + c(x',y') \le c(x',y) + c(x,y'), 
\end{equation}
for all $x \le x'$, $y \le y'$,
the optimum transportation map corresponds to the Hoeffding distribution \cite{rachev1985monge} defined as follows. Let $C_X(x)$ and $C_Y(y)$ be the cumulative distributions of $X$ and $Y$ respectively, and let $C_{XY}(x,y)$ be the cumulative transportation map, that is:
\begin{equation}
\label{eq.cumulativeMap}
    C_{XY}(x,y) = \int_{-\infty}^x \! \int_{-\infty}^y S_{XY} (u,v) du dv.
\end{equation}
The optimum transportation map is obtained by letting:
\begin{equation}\label{HoeffdingDistr}
C_{XY}^*(x,y) = \min\{C_X(x), C_Y(y)\}, \quad \forall(x,y) \in \mathds{R}^2,
\end{equation}
which generalizes the {\em NWC} rule. Given the optimum transportation map, one can compute $\SS \MM(P_Y,P_X)$ by evaluating the integral in (\ref{eq.Monge_Kant}).In general, however,  finding a closed form expression is not an easy task.

A particularly simple and insightful formula can be obtained when the cost function corresponds to the squared Euclidean distance. Let us assume, then, that $c(x,y) = (x-y)^2$ (in this case $\SS\MM(P_X,P_Y)$ corresponds to the squared Mallows distance - see equation (\ref{Mallow_dist}) -  and let $X$ and $Y$ be two continuous sources with means $\mu_X$ and $\mu_Y$, variances $\sigma_X$ and $\sigma_Y$ and covariance $covXY$. As shown in \cite{IrpinoR} (decomposition theorem), the expectation in (\ref{Mallow_dist}) can be rewritten as follows:
\begin{align}
\label{decomposition_theorem}
E_{XY} [(X - Y)^2]  = &  (\mu_X - \mu_Y)^2 + (\sigma_X - \sigma_Y)^2\\ \nonumber
& + 2[\sigma_X\sigma_Y - covXY],
\end{align}
where the three terms express, respectively, the difference in location, spread and shape between the variables $X$ and $Y$ \cite{kovsmeljmallows}. Interestingly, the covariance $covXY$ is the only term in (\ref{decomposition_theorem}) which depends on the joint pdf of $X$ and $Y$. Then, in order to find the security margin, we only have to compute the maximum covariance over all the possible joint pdf's:
%
\begin{align}
\SS\MM_{L_2^2}(P_X, P_Y) = & (\mu_X - \mu_Y)^2 + (\sigma_X - \sigma_Y)^2\\ \nonumber
&  + 2[\sigma_X\sigma_Y - \max_{P_{XY}:\tiny{\substack{\sum_y P_{XY} = P_X \\ \sum_x P_{XY} = P_Y}}} covXY].
\label{covariance}
\end{align}
By assuming that $X$ and $Y$ are independent, i.e. $P_{XY} = P_X P_Y$, we have $covXY =0$, hence permitting us to derive a general upper bound for the security margin:\footnote{We point out that relation (\ref{decomposition_theorem}), as well as the upper bound in (\ref{SMupperbound}), holds for the discrete case too.}
\begin{equation}
\label{SMupperbound}
\SS\MM_{L_2^2}(P_X,P_Y) \le (\mu_X - \mu_Y)^2 + \sigma^2_X + \sigma^2_Y.
\end{equation}
When $P_X$ and $P_Y$ have the same form, for instance when the random variables $X$ and $Y$ are both distributed according to a Gaussian or a Laplacian distribution, the security margin assumes a particularly simple expression. In this case, in fact, it is possible to turn $P_X$ into $P_Y$ by imposing a deterministic relationship between $X$ and $Y$, namely $Y = \frac{\sigma_Y}{\sigma_X} X + (\mu_Y - \frac{\sigma_Y}{\sigma_X} \mu_X)$. In this way the covariance term is maximum and equal to $\sigma_X \sigma_Y$, and hence the contribution of the shape term in the security margin vanishes, yielding:
\begin{equation}
{\SS \MM}_{L_2^2}(P_X, P_Y) =  (\mu_X - \mu_Y)^2 + (\sigma_X - \sigma_Y)^2.
\label{sm_continuous_same_class}
\end{equation}
This is a remarkable, and somewhat surprising, result stating that the distinguishability of two sources belonging to the same class depends only on their means and variances, regardless of their particular pdf.

\section{The Security Margin with $L_{\infty}$ distance}
\label{sec.maxdiff}

We conclude the paper by extending the definition of the Security Margin to the case in which the distortion measure constraining the attacker is expressed in terms of the maximum absolute distance between the samples of $y^n$ and $z^n$, that is to the case in which the distortion is measured by relying on the $L_{\infty}$ distance.

The interest in this case is motivated by the importance that the $L_{\infty}$ distance has in applications where the perceptual distortion between the sequence $y^n$ and the attacked sequence $z^n$ must be taken into account. This is the case, for instance, in image forensics applications \cite{Boh12, BarniMMSEC, BarniIJDCF, BT12,FernPedro2013}, wherein the attacker is interested in hiding the true source of an image. In this case, the use of a distortion measure based on the $L_{\infty}$ distance ensures that the attacked image is perceptually similar to the original one. In our analysis, we will refer to the case of known sources, the extension to the $SI_{tr}$ case being immediate.

\subsection{The $SI_{ks}$ game with $L_{\infty}$ distance}

We start by observing that the adoption of the $L_{\infty}$ distance requires that the $SI_{ks}$ game is, partly, redefined due to the non-additive nature of the distortion constraint. In this case, in fact, it does not make any sense to define the distortion constraint in terms of average per-letter distortion and let the overall allowed distortion to increase with $n$.


Similarly to the previous cases, it is possible to express the distortion constraint by limiting the set of transportation maps the attacker can choose from. More specifically, we observe that the maximum distance between the two sequences $y^n$ and $z^n$ can  be rewritten as follows:
\begin{equation}
\label{link_L_infty_TPmap}
 d_{L_\infty}(y^n, z^n) = \max_j |z_j - y_j| =  \max_{(i,j) : S_{YZ}^n(i,j) \neq 0} |i-j|.
\end{equation}
By using the above formula in the definition of the set of admissible maps (i.e. in the second line of \eqref{eq.admissiblemap1}), we can still define the set of strategies of the attacker as the set of rules associating an admissible map to the to-be-attacked sequence, as in \eqref{eq.SA_TR_as}. In the following, we will refer to the set of admissible maps resulting from the use of the $d_{L_\infty}$ distance as $\A_{L_\infty}^n(L_{max}, P_{y^n})$.

Passing to the analysis of the indistinguishability region, it is easy to see that relation (\ref{eq.indistinguib}) continues to hold by replacing $\A^n(L_{max}, P_{y^n})$ with $\A_{L_\infty}^n(L_{max}, P_{y^n})$. In fact, the dominant strategy for the defender does not depend on the set of strategies available to the attacker. The asymptotic version of  $\Gamma_{L_\infty}^n(P_X, \lambda, L_{max})$ can also be defined as in (\ref{eq.indistinguib_limit}), namely:
\begin{align}
\label{eq.indistinguib_infty}
    \Gamma_{L_{\infty}}(P_X, \lambda, L_{max}) &= \\ \nonumber
    \{P \in \PP : \exists ~ S_{YZ} & \in \A_{L_{\infty}}(L_{max}, P) \text{ s.t. } S_{Z} \in \Lambda^{*}(P_X, \lambda)\},
\end{align}
where $\A_{L_{\infty}}(L_{max}, P)$ is the asymptotic counterpart of $ \A^n_{L_{\infty}}(L_{max}, P)$. The next step requires the extension of Theorem \ref{theorem2} to the $SI_{ks}$ game with $L_{\infty}$ distance, that is we need to prove that the set in (\ref{eq.indistinguib_infty}) contains all the sources that can not be distinguished from $X$ because of the attack, even when the length of the observed sequence tends to infinity. This is a critical step since such theorem was proved in \cite{CandT} by assuming an additive distortion measure, which clearly is not the case when the $L_{\infty}$ distance is adopted. Roughly speaking, we need to prove that when $n \rarrow \infty$ the elements of $\Gamma_{L_{\infty}}^n(P_X, \lambda, L_{max})$ are dense in $\Gamma_{L_{\infty}}(P_X, \lambda, L_{max})$ (in which case Theorem \ref{theorem2} can be proven in a way similar to Sanov's Theorem \cite{CandT}). More formally, we need to prove that for any $P_Y \in \Gamma_{L_{\infty}}(P_X, \lambda, L_{max})$ and any $\delta > 0$, a pmf $Q^n \in \Gamma^n_{L_{\infty}}(P_X, \lambda, L_{max})$ exists such that the distance between $P_Y$ and $Q^n$ is smaller than $\delta$. The proof requires only some minor modifications with respect to the proof given in \cite{BT13} (Lemma 2 in the Appendix) and is skipped for sake of brevity.

\subsection{Security Margin for the $SI_{ks}$ game with $L_{\infty}$ distance}

As a next step, we must study the behavior of the indistinguishability region of the test when $\lambda \rarrow 0$ (to determine the smallest indistinguishability region). As we will see, even if the adoption of the $d_{L_{\infty}}$ distance prevents a direct formulation of the problem in terms of {\em EMD}, the distinguishability between two sources $X$ and $Y$ is still closely related to the optimal transportation map between  $P_X$ and $P_Y$. The basis for such a connection is rooted in the following property.

\begin{property}
\label{property_Huffman_dinfty}
Given two distributions $P$ and $Q$, the transportation map $S_{PQ}^{\text{\em NWC}}$ obtained by applying the {\em NWC} rule to P and Q is a solution of the problem
\begin{equation}
\label{relation_d_infty_NWC}
\min_{S_{YZ}: S_Y = P, S_Z = Q} \left(\max_{(i,j) \in S_{YZ}(i,j) \neq 0} |i - j|\right).
\end{equation}
\end{property}
\begin{proof}
Let $S^* \ne S_{PQ}^{\text{\em NWC}}$ be a generic transformation mapping $P$ into $Q$. Given that $S^* \ne S_{PQ}^{\text{\em NWC}}$ there exists at least one quadruple of bins $(t,r,v,s)$, with $t < r$ and $v < s$, for which, $S^*(t,s) > 0$ and $S^*(r,v) > 0$. Let us assume, without loss of generality, that $S^*(t,s) \le S^*(r,v)$.
We now define a new map $S'$ which is obtained from $S^*$ by letting:
\begin{align}
\label{eq.Sprimo}
     S'(t,v) &= S^*(t,v) + S^*(t,s) \\ \nonumber
     S'(t,s) &= 0 \\ \nonumber
     S'(r,v) &= S^*(r,v) - S^*(t,s) \\ \nonumber
     S'(r,s) &= S^*(r,s) + S^*(t,s).
\end{align}
Since $\max\{|t - s|, |r - v|\} > \max\{|t - v|, |r - s|\}$, the maximum distortion introduced by $S'$ is lower than or equal to that introduced by $S^*$, that is:
\begin{equation}
\label{any_map_vs_NWC}
\max_{(i,j) \in S^*(i,j) \neq 0} |i - j| \ge \max_{(i,j) \in S'(i,j) \neq 0} |i - j|.
\end{equation}
We now inspect $S'$, if there is another quadruple of bins $(t',r',v',s')$ satisfying the same properties of $(t,r,v,s)$, we let $S^* = S'$ and iterate the above procedure. The process ends when no quadruple of bins with the required properties exists and hence when $S' = S_{PQ}^{\text{\em NWC}}$. Since at each step the distortion introduced by the new map does not increase, the above procedure proves that $S_{PQ}^{\text{\em NWC}}$ introduces a distortion lower than or equal to that introduced by any other $S^*$ mapping $P$ into $Q$, thus proving that $S_{PQ}^{\text{\em NWC}}$ achieves the minimum in \eqref{relation_d_infty_NWC}.
\end{proof}
Thanks to Property \ref{property_Huffman_dinfty}, the set $\Gamma_{L_{\infty}}(P_X, \lambda, L_{max})$ in \eqref{eq.indistinguib_infty} can be rewritten as follows:
\begin{align}
\label{eq.indistinguib_EMD_Linfty}
    \Gamma_{L_{\infty}}(P_X, \lambda, L_{max}) &  =
    \{P \in \PP :  ~\exists ~ Q \in {\Lambda^*}(P_X, \lambda) ~ \text{s.t} ~  \\ &\max_{(i,j) : S_{PQ}^{\text{\em NWC}}(i,j) \neq 0} |i - j| \le L_{max} \}.\nonumber
\end{align}
By letting $\lambda$ tend to 0, we obtain the smallest indistinguishability region, thus extending Theorem \ref{theorem_EMD} to the $SI_{ks}$ game with $d_{L_{\infty}}$ distance.
\begin{theorem}
\label{theorem_EMD_Linfty}
Given two sources $X \sim P_X$ and $Y \sim P_Y$ and a maximum allowable per-letter distortion $L_{max}$, and given:
\begin{equation}
\Gamma(P_X,L_{max}) = \{P \in \PP: \max_{(i,j) \in S_{P P_X}^{\text{\em NWC}}} |i - j| \le L_{max}\},
\label{Gamma_X_infty}
\end{equation}
the maximum achievable false negative error exponent $\varepsilon$ for the $SI_{ks}$ game with $L_{\infty}$ distance is
\begin{equation}
\lim_{\lambda \rarrow 0} \lim_{n \rarrow \infty} - \frac{1}{n} \log P_{fn} = \min_{P \in \Gamma_{L_{\infty}}(P_X,L_{max})} \DD(P || P_Y).
\label{best_e_e_appendix}
\end{equation}
\end{theorem}
\begin{proof}
The proof relies on the extension  of Property \ref{prop.inclusion} and Lemma \ref{lemma_A} to the $L_{\infty}$ case. The extension of Property \ref{prop.inclusion} is immediate since, once again, the indistinguishability region depends on $\lambda$ only through ${\Lambda^*}(P_X, \lambda)$, whose form does not depend on the particular norm adopted to express the distortion constraint. The extension of Lemma \ref{lemma_A} requires some more care and is proven in Appendix \ref{app.lemma_A_ext}. For the rest, the theorem can be proven by reasoning as in the proof of Theorem \ref{theorem_EMD}.
\end{proof}
As a consequence of Theorem \ref{theorem_EMD_Linfty}, the distinguishability of two sources depends again on the optimum transportation map between the pmf's of the two sources. Specifically, given the sources $X$ and $Y$, the defender is able to distinguish between them in this adversarial setting, only if
\begin{equation}
\label{eq.condLinf}
\max_{(i,j) \in S^{\text{\em NWC}}_{P_Y P_X}} |i - j| > L_{max}.
\end{equation}
Condition \eqref{eq.condLinf} can be used to determine the maximum attacking distortion for which D is able to distinguish the sources $X$ and $Y$, i.e. $\SS\MM(P_X,P_Y)$.
\begin{definition}[Security Margin for the $L_{\infty}$ case]
Let $X \sim P_X$ and $Y \sim P_Y$ be two discrete memoryless sources. The maximum distortion for which the two sources can be reliably distinguished in the $SI_{ks}$ setting with $L_{\infty}$ distance is given by
\begin{equation}
\SS \MM_{L_{\infty}}(P_Y, P_X) = \max_{(i,j) : S^{\text{\em NWC}}_{P_Y P_X}(i,j) \neq 0} |i-j|,
\label{security_margin_discrete}
\end{equation}
where $S^{\text{\em NWC}}_{P_Y P_X}$ is obtained by applying the {\em NWC} rule to map $P_Y$ into $P_X$.
\end{definition}

Even if we proved Theorem \ref{theorem_EMD_Linfty} for the case of known sources, it is possible to extend it to the $SI_{tr}$ game. The proof goes along the same lines used for the $SI_{ks}$ case and is omitted for sake of brevity.

\section{Conclusions}
\label{sec.conc}

By interpreting the attacker's optimum strategy in the $SI_{ks}$ (and $SI_{tr}$) game as the solution of an optimum transport problem, we introduced the concept of security margin, a single measure summarizing the distinguishability of two sources under adversarial conditions. We also described an efficient algorithm to compute the security margin between several classes of sources. By relying on the security margin concept, we can understand who between the attacker and the defender is going to win the source identification game under asymptotic conditions. Among the practical applications of our analysis we mention image forensics, wherein the defender is interested in distinguishing images produced by different devices, and intrusion detection, in which the defender is willing to distinguish normal and anomalous behaviors. In the first case, knowing the $\SS\MM$ between the statistics ruling the emission of images from different sources permits to compute the amount of distortion required to make the images produced by the two sources indistinguishable. In the latter case, the $\SS \MM$ determines how much an intruder must deviate from the intended, anomalous, behavior to make its presence undetectable by the analyst.

\section*{Acknowledgment}
We thank Alessandro Agnetis for the useful discussions on the optimization problems underlying the computation of the {\em EMD}.

This work has been partially supported by the European Office of Aerospace Research and Development under Grant FA8655-12-1- 2138: AMULET - A multi-clue approach to image forensics, and the the REWIND Project, funded by the Future and Emerging Technologies (FET) programme within the 7FP of the EC, under grant 268478.

  \bibliographystyle{IEEEtran}
\bibliography{TIT_SM}

\numberwithin{equation}{section}
\appendix

\renewcommand{\theequation}{A\arabic{equation}}

\subsection{Behavior of $\Gamma(P_X, \lambda,L_{max})$ and $\Gamma_{tr}(R, \lambda,L_{max})$ for $\lambda \rarrow 0$.}
\label{app.lemma_A}

We start by studying the behavior of $\Gamma(P_X, \lambda,L_{max})$  when $\lambda \rarrow 0$. More specifically, we show that for small values of $\lambda$ the set $\Gamma(P_X, \lambda,L_{max})$ approaches $\Gamma(P_X, L_{max})$ smoothly.

As a first step, we highlight the following property.
\begin{property}
\label{property_EMD}
{\em EMD}$(P,Q)$ is a continuous and convex function of $P$ and $Q$.
\end{property}
\begin{proof}
Property \ref{property_EMD} follows immediately if we look at the {\em EMD} as the solution of a Linear Programming (LP) problem (see Section \ref{sec.Huffman_NWC}), wherein $P$ and $Q$ are the known terms of the linear constraints. In fact, it is a known result in operations research that the minimum of the objective function of an LP problem is a continuous and convex function of the known terms of the linear constraints \cite{LinearOpt}.
\end{proof}

By exploiting the continuity of the divergence and the continuity and convexity of the {\em EMD}, we now show that when $\lambda$ tends to 0, the set $\Gamma(P_X, \lambda, L_{max})$ tends to  $\Gamma(P_X, L_{max})$ regularly. More precisely, the following lemma holds.
\begin{lemma}
\label{lemma_A}
Let $X \sim P_X$ be an information source and $L_{max}$ the maximum allowable average per-letter distortion in the $SI_{ks}$ game. The set $\Gamma(P_X,\lambda,L_{max})$, defined in \eqref{eq.indistinguib_EMD}, satisfies the following property:
\begin{align}
\forall \tau > 0, & \exists \lambda > 0   \text{ s.t.  }  \forall P \in \Gamma(P_X, \lambda, L_{max})\\
 & \exists P' \in \Gamma(P_X, L_{max}) \text{ s.t.  } P \in B(P', \tau),\nonumber
\label{lemma_A_rel}
\end{align}
where $\Gamma(P_X,L_{max})$ is defined as in \eqref{Gamma_X} and  $B(P', \tau)$ is a ball centered in $P'$ with radius $\tau$.
\end{lemma}
\begin{figure}[t!]
\centering \includegraphics[width = 0.75\columnwidth]{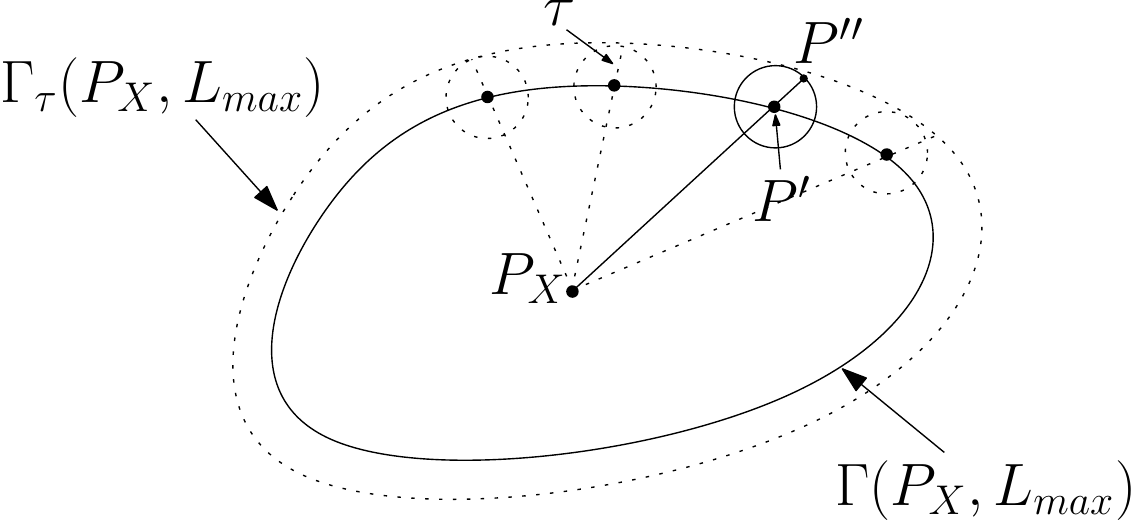}
\caption{Graphical representation of the set $\Gamma_{\tau}(P_X, L_{max})$.}
\label{fig.sketch_set}
\end{figure}
\begin{proof}
\noindent Throughout the proof we will refer to Figure \ref{fig.sketch_set} where all the sets and quantities involved in the proof are sketched. For any $\tau > 0$, we consider the set:
\begin{align}
&\Gamma_{\tau}(P_X,L_{max}) = \\ \nonumber
& \{P : \exists P' \in \Gamma(P_X, L_{max}) \text{  s.t.  }  P \in B(P', \tau)\}.
\end{align}
With such a definition, we can rephrase (\ref{lemma_A_rel}) as follows:
\begin{align}
\label{lemma_A_rel2}
\forall \tau > 0, \exists  \lambda > 0   \text{ s.t.} \text{  } \Gamma(P_X, \lambda, L_{max})\subseteq \Gamma_{\tau}(P_X, L_{max}).
\end{align}
%

For sake of simplicity, we will prove a slightly stronger version of the lemma by means of the following two-step proof. First, we will show that a subset of $\Gamma_{\tau}(P_X,L_{max})$ exists having the following form:
\begin{equation}
\Gamma_{\tau}^{sub}(P_X,L_{max}) = \{P : \text{\em EMD}(P,P_X) \le L_{max} + \delta(\tau)\},
\end{equation}
for some $\delta(\tau) > 0$. Then, we will prove that for small enough $\lambda$, any $P \in \Gamma(P_X, \lambda, L_{max})$ belongs to $\Gamma_{\tau}^{sub}(P_X,L_{max})$.\\
To start with, let $P'$ be any point on $\mathcal{B}(\Gamma(P_X,L_{max}))$, the boundary of $\Gamma(P_X,L_{max})$. Among all the the points on the boundary of the ball of radius $\tau$ and centered in $P'$, consider the one, name it $P''$, lying along the direction given by the line joining $P_X$ and $P'$ and falling outside $\Gamma(P_X,D_{max})$ (see Figure \ref{fig.sketch_set}). By the convexity of the {\em EMD} (Property \ref{property_EMD}) and since {\em EMD = 0} if and only if $P=P_X$, we conclude that {\em EMD}$(P'',P_X) > $ {\em EMD}$(P',P_X)$. Since $P'$ lies on the boundary of $\Gamma(P_X,L_{max})$ we know that {\em EMD}$(P'',P_X) = L_{max} + \mu$, where $\mu = \mu(P',\tau)$ is a strictly positive quantity. We now show that the first part the proof holds by letting $\delta(\tau) = \min_{P' \in \mathcal{B}(\Gamma(P_X,L_{max}))} \mu(P',\tau)$.
To this purpose, let $P$ be any point in set $\Gamma^{sub}_{\tau}(P_X,L_{max})$ for the above choice of $\delta(\tau)$.
If $P \in \Gamma(P_X,L_{max})$, then, by definition, $P$ also belongs to $\Gamma_{\tau}(P_X,L_{max})$. On the other side, if $P$ lies outside $\Gamma(P_X,L_{max})$, let us denote by $P^*$ the point lying on the boundary of the set $\Gamma(P_X,L_{max})$ along the line joining $P$ and $P_X$, and let $P^{**}$ be the point where the same line crosses the ball $B(P^*,\tau)$ outside $\Gamma(P_X,L_{max})$. Now, {\em EMD}$(P,P_X) \le L_{max} + \delta(\tau) \le$ {\em EMD}$(P^{**},P_X)$ by construction. Because of the convexity of {\em EMD}, then $P \in B(P^*,\tau)$ as required.\\
Let us now pass to the second part of the proof.
First, we notice  that set $\Gamma(P_X,\lambda, L_{max})$ depends on $\lambda$ only through the acceptance region $\Lambda^*(P_X,\lambda)$. If $\lambda$ is small, due to the continuity of the divergence, for any $Q \in \Lambda^*(P_X,\lambda)$ we will have $Q \in B(P_X,\kappa(\lambda))$ for some $\kappa(\lambda)$ such that $\kappa(\lambda) \rarrow 0$ when $\lambda \rarrow 0$. Let, then, $P$ be a pmf in $\Gamma(P_X,\lambda, L_{max})$. By definition, a $Q \in \Lambda^*(P_X,\lambda)$ exists s.t. {\em EMD}$(P,Q) \le L_{max}$. If $\lambda$ is small, due to the proximity of $Q$ to $P_X$ and the continuity of the {\em EMD} we have that {\em EMD}$(P,P_X) < ${\em EMD}$(P,Q) + \eta(\lambda) \le L_{max} + \eta(\lambda)$ with $\eta(\lambda)$ approaching 0 when $\lambda \rarrow 0$. In particular, if $\lambda$ is small enough $\eta(\lambda) < \delta(\tau)$ and hence $P \in\Gamma^{sub}_{\tau}(P_X, L_{max})$ which in turn is entirely contained in $\Gamma_{\tau}(P_X, L_{max})$ thus completing the proof.
\end{proof}

In the same way, we can prove that Lemma \ref{lemma_A} holds also when $\Gamma(P_X, \lambda,L_{max})$ is replaced by $\Gamma_{tr}(R, \lambda,L_{max})$ and $\Gamma(P_X, L_{max})$ by $\Gamma(R, L_{max})$ with a generic $R$ instead of $P_X$.
To be convinced about that, it is sufficient to note that the only difference between $\Gamma$ and $\Gamma_{tr}$ relies on the test function which defines the acceptance region, respectively the divergence and the $h_c$ function. Since the $h_c$ function is still a continuous and convex function and, likewise $\DD$, is equal to zero if and only if its arguments are identical, the proof that we used for Lemma \ref{lemma_A} still holds.

\subsection{Behavior of $\Gamma_{L_{\infty}}(P_X, \lambda,L_{max})$ for $\lambda \rarrow 0$.}
\label{app.lemma_A_ext}
We prove that when $\lambda \rarrow 0$, $\Gamma_{L_{\infty}}(P_X, \lambda,L_{max})$  approaches $\Gamma_{L_{\infty}}(P_X, L_{max})$ regularly, in the sense stated by the following lemma.
\begin{lemma}[Extension of Lemma \ref{lemma_A} to the $L_{\infty}$ case]
Let $X \sim P_X$ be an information source and $L_{max}$ the maximum per-sample distortion allowed to the attacker. The set $\Gamma_{L_{\infty}}(P_X, \lambda, L_{max})$, defined in Section \ref{sec.maxdiff}, satisfies the following property:
\begin{align}
\label{lemma_A_rel_inf}
\forall \tau > 0, \exists & \lambda > 0   \text{ s.t.},  \text{  }  \forall P \in \Gamma_{L_{\infty}}(P_X, \lambda, L_{max})\\
 & \exists P' \in \Gamma_{L_{\infty}}(P_X, L_{max}) \text{ s.t. } P \in B(P', \tau),\nonumber
\end{align}
where $B(P', \tau)$ is a ball centered in $P'$ with radius $\tau$.
\end{lemma}

\begin{proof}
We will prove the lemma by assuming that the distance defining the ball $B(P', \tau)$ is the $L_1$ distance, extending the proof to other distances being straightforward.

For a fixed $\tau > 0$, let $P$ be a pmf in  $\Gamma_{L_{\infty}}(P_X, \lambda, L_{max})$ for some $\lambda$. This means that at least one pmf $Q \in \Lambda^*(P_X, \lambda)$ exists, such that $P$ can be mapped into $Q$ with maximum shipment distance lower than or equal to $L_{max}$. From equation \eqref{eq.lambda_limit} and by exploiting the continuity of the divergence function, we argue that $Q \in \mathcal{B}(P_X, \gamma(\lambda))$ for some positive $\gamma(\lambda)$, and where $\gamma(\lambda) \rarrow 0$ as $\lambda \rarrow 0$. Accordingly, $P_X$ can be written as $P_X(j) = Q(j) + \gamma(j)$, $\forall j$, where $\sum_{j \in \XX} |\gamma(j)| < \gamma(\lambda)$. Note that, by construction, $\sum_j \gamma(j) = 0$ and $\gamma(j) \rarrow 0$ when $\lambda \rarrow 0$. Let $S_{PQ}$ be an admissible map bringing $P$ into $Q$ (such a map surely exists by construction). We prove the lemma by explicitly building a pmf $P'$ and a new admissible transportation map $S'$, such that, $P'$ is arbitrarily close to $P$ (for a small enough $\lambda$) and $S'$ maps $P'$ into $P_X$.
We start by introducing two new quantities, namely $\gamma^+(j)$, defined as follows:
\begin{align}
\label{eq.gammapiu}
    \gamma^+(j) & = \gamma(j)  &\text{if } P_X(j)-Q(j) \ge 0 \\ \nonumber
    \gamma^+(j) & = 0  &\text{if } P_X(j)-Q(j) < 0,
\end{align}
and $\gamma^-(j)$ defined as
\begin{align}
\label{eq.gammameno}
    \gamma^-(j) & = -\gamma(j) &\text{if } P_X(j) - Q(j) < 0 \\ \nonumber
    \gamma^-(j) & = 0  &\text{if } P_X(j) - Q(j) \ge 0.
\end{align}
A graphical interpretation of $\gamma^+$ and $\gamma^-$ is given in Figure \ref{fig.last}. Clearly, $\sum_j \gamma^-(j) = \sum_j \gamma^+(j)$.
\begin{figure}[t!]
\centering \includegraphics[width = 0.65\columnwidth]{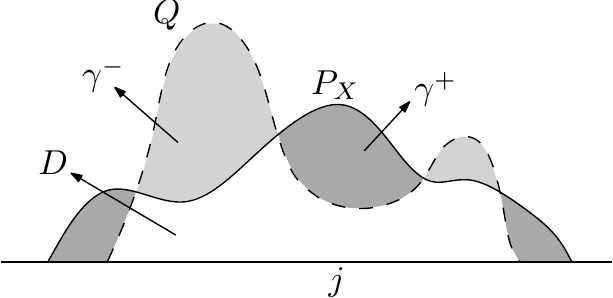}
\caption{Geometric interpretation of $\gamma^+$, $\gamma^-$ and $D(j)$.}
\label{fig.last}
\end{figure}
With the above definitions, we can look at the demand distribution $Q$ as consisting of two amounts: the mass distribution $D$, with $D(j) = \min\{P_X(j),Q(j)\}$, and $\gamma^-$. According to the superposition principle, the map $S_{PQ}$ can then be split into two sub-maps: one which satisfies the demand of $D$ (let us call it $S_{PQ}^D$), and one that satisfies the demand of $\gamma^-$ (let us call it $S_{PQ}^{\gamma}$). The same distinction can be made in the source distribution, as follows:
\begin{align}
P(i) = & \sum_{j} S_{PQ}(i,j)\\
 = & \sum_j S_{PQ}^D(i,j)  + \sum_j S_{PQ}^{\gamma}(i,j) = P_D(i) + P_{\gamma}(i), \nonumber
\end{align}
where $P_D$ and $P_{\gamma}$ are the masses in the source distribution which are used to satisfy the mass demand pertaining to $D$ and $\gamma^-$ according to mapping $S_{PQ}$. Then, $\sum_i P_D(i) = D$ and $\sum_i P_{\gamma}(i) = \gamma^-$.
In order to construct the pmf $P'$ we are looking for, we simply remove from $P$ the amount of mass $P_{\gamma}$ used to fill $\gamma^-$ and redistribute it according to $\gamma^+$. Specifically, we have
\begin{align}
& P'(i) = P_D(i) + \gamma^+(i) \\
& S'(i,j) = S_{PQ}^D(i,j) + \gamma^+(j) \delta(i,j),
\end{align}
where $\delta(i,j)$ is equal to 1 if $i=j$ and 0 otherwise. It is easy to see that applying the transportation map $S'(i,j)$ to $P'$ yields $P_X$.
Besides, from the procedure adopted to build $S'$, it is evident that
\begin{equation}
\label{eq.distconst}
    \max_{(i,j): S'(i,j) \ne 0} |i-j| \le \max_{(i,j): S_{PQ}(i,j) \ne 0} |i-j| \le L_{max},
\end{equation}
(the only new shipments introduced are from a bin to itself).
In addition, the distance between $P'$ and $P$ is, by construction, lower than $\gamma(\lambda)$, which can be made arbitrarily small by decreasing $\lambda$, thus completing the proof of the lemma.
\end{proof}

\end{document}